%% file: paper.tex
\documentclass[11pt, letterpaper]{article}

\usepackage{wrapfig}

\input{envs}

\newcommand{\os}{\textrm{\sc COSSP}\xspace}
\newcommand{\pc}{\textrm{\sc PCSP}\xspace}
\newcommand{\edf}{\textrm{\sc EDF}\xspace}

\newcommand{\prc}{\textrm{\sc PR2C}\xspace}
\newcommand{\hcp}{\textrm{\sc HCCP}\xspace}
\newcommand{\gmc}{\textrm{\sc GMCC}\xspace}
\newcommand{\dks}{\textrm{\sc DkS}\xspace}
\newcommand{\tc}{\tilde{c}}
\newcommand{\td}{\tilde{d}}
\newcommand{\pt}{\mathcal{P}}
\newcommand{\rs}{\mathcal{R}}
\newcommand{\lps}{\mathcal{S}}

\begin{document}

\title{ Flow-time Optimization For Concurrent Open-Shop and Precedence Constrained Scheduling Models}
\author{ Janardhan Kulkarni \thanks{Microsoft Research, Redmond,  {\tt jakul@microsoft.com} } \and
	Shi Li \thanks{University at Buffalo, Buffalo, NY,  {\tt shil@buffalo.edu}. The work is in part supported by NSF grants CCF-1566356 and CCF-1717134.}}

\date{}
\maketitle

\thispagestyle{empty}

\begin{abstract}
Scheduling a set of jobs over a collection of machines is a fundamental problem that needs to be solved millions of times a day in various computing platforms: in operating systems, in large data  clusters, and in data centers. 
Along with makespan, {\em flow-time}, which measures the length of time a job spends in a system before it completes,  is arguably the most important metric to measure the performance of a scheduling algorithm. In recent years, there has been a remarkable progress in understanding  flow-time based objective functions in diverse settings such as unrelated machines scheduling, broadcast scheduling,  multi-dimensional scheduling, to name a few.  

Yet, our understanding of the flow-time objective is limited mostly to the scenarios where jobs have simple structures; in particular, each job is  a single self contained entity. On the other hand, in almost all real world applications, think of MapReduce settings for example,  jobs have more complex structures. In this paper,  we consider two classical scheduling models that capture complex job structures: 1) {\em concurrent open-shop scheduling} (\os) and  2) {\em precedence constrained scheduling} (\pc). Our main motivation to study these problems specifically comes from their relevance to two scheduling problems that have gained importance in the context of data centers: {\em co-flow scheduling} and {\em DAG scheduling}.  
We design almost optimal approximation algorithms for \os and \pc, and show hardness results.

\end{abstract}

\input{intro.tex}
\input{ourtechniques.tex}

\input{openshop.tex}
\input{precedence.tex}

\section{Conclusion and Open Problems}
In this paper, we considered two classical scheduling models and gave nearly optimal approximation algorithms. 
However, plenty of interesting open problems remain. 
In the open-shop scheduling model, the most interesting open question is if we can generalize our result to the non-concurrent case. 
Another interesting open problem is to obtain a non-trivial  approximation algorithm for the co-flow scheduling problem.
For the precedence constrained scheduling model, a more open-ended question is if we can show results without any speed augmentation for some interesting special cases.

\bibliographystyle{plainnat}
\bibliography{mmf}

\end{document}

%% file: envs.tex
\usepackage{times}
\usepackage{latexsym}
\usepackage{amsfonts,amsthm,amssymb}
\usepackage{amsmath}
\usepackage{euscript}
\usepackage{amstext}
\usepackage{graphicx}
\usepackage{color}
\usepackage{url}
\usepackage{verbatim}
\usepackage{xspace}
\usepackage{framed}
\usepackage{multirow}
\usepackage[numbers,sort&compress]{natbib} 
\usepackage[dvipsnames,usenames,table]{xcolor}
\usepackage[backref,colorlinks=true,urlcolor=blue,linkcolor=RoyalBlue,citecolor=OliveGreen]{hyperref}

\newtheorem{lemma}{Lemma}[section]
\newtheorem{theorem}[lemma]{Theorem}

\newtheorem{definition}[lemma]{Definition}
\newtheorem{corollary}[lemma]{Corollary}

\newtheorem{claim}[lemma]{Claim}
\newtheorem{remark}[lemma]{Remark}

{\hspace*{\fill}$\Box$\par}

	\newcommand{\initOneLiners}{%
		\setlength{\itemsep}{0pt}
		\setlength{\parsep }{pt}
		\setlength{\topsep }{0pt}
	}

	\renewcommand{\vec}[1]{\mathbf{#1}}

\usepackage{fullpage}
\usepackage{amsthm,amsmath,amssymb}
\usepackage{algorithm, algorithmic}
\usepackage{mdframed}
\usepackage{graphicx}
\usepackage{enumitem}
\usepackage{xspace}
\usepackage{mdframed}
\usepackage{multirow}
\usepackage{hhline}
\usepackage{todonotes}

\makeatletter
\renewcommand{\paragraph}{%
	\@startsection{paragraph}{4}%
	{\z@}{1.25ex \@plus 1ex \@minus .2ex}{-1em}%
	{\normalfont\normalsize\bfseries}%
}
\makeatother

\newcommand{\set}[1]{\left\{#1\right\}}

\newcommand{\floor}[1]{\left\lfloor#1\right\rfloor}
\newcommand{\ceil}[1]{\left\lceil#1\right\rceil}

\renewcommand{\tilde}{\widetilde}

\newcommand{\sfd}{{\mathsf{d}}}

%% file: intro.tex

\section{Introduction}
	\label{sec:intro}

Scheduling a set of jobs over a collection of machines is a fundamental problem that needs to be solved millions of times a day in various computing platforms: in operating systems, in large data  clusters, and in data centers. 
Along with makespan, {\em flow-time}, which measures the length of time a job spends in a system before completing,  is arguably the most important metric to measure the performance of a scheduling algorithm. In recent years, there has been a remarkable progress in understanding  flow-time related objective functions in diverse settings such as unrelated machines scheduling \cite{GargK07, ImM11, AnandGK12, ImKMP14, BansalK15}, broadcast scheduling \cite{bansal2014better, ImM12, BansalKN10},  multi-dimensional scheduling \cite{ImKM15, ImKM14}, to name a few.  

Yet, our understanding of the flow-time based objective functions is mostly limited to the scenarios where jobs have simple structures; in particular, each job is  a single self contained entity.
On the other hand, in almost all real world applications, jobs have more complex structures. 
Consider the MapReduce model for example. 
Here, each job consists of a set of {\em Map} tasks and a set of {\em Reduce} tasks.
Reduce tasks cannot be processed unless Map tasks are completely processed\footnote{In some MapReduce applications, Reduce tasks can begin after the completion of a subset of Map tasks.}.
A MapReduce job is complete only when all Map and Reduce tasks are completed.
Motivated by these considerations, in this paper, we consider  two classical  scheduling models  that capture more complex job structures:  1) {\em concurrent open-shop scheduling (\os)}, and 2)  {\em precedence constrained scheduling (\pc)}. Our main reason to study these problems specifically comes from their relevance to two scheduling problems that have gained importance in the context of data centers: {\em co-flow scheduling and DAG scheduling}. 
We discuss more about how these problems relate to \os and \pc in Section \ref{apps}.

The objective function we consider in this paper is minimizing the {\em sum of general delay costs of jobs}, first introduced in an influential paper by Bansal and Pruhs \cite{BansalP10} in the context of single machine scheduling. In this objective, for each job $j$ we are given a {\em non-decreasing} function $g_j:\mathbb{Z_+} \rightarrow \mathbb{Z_+}$, which gives the cost of completing the job at time $t$. The goal is to minimize $\sum_{j} g_j(C_j)$, where $C_j$ is the completion time of $j$.
A desirable aspect of the general delay cost functions is that they capture several widely studied flow-time and completion time based objective functions. 
\begin{itemize}[topsep=3pt,itemsep=3pt, parsep=0pt]
	\item Minimizing the sum of weighted flow-times of jobs. This is captured by the function $g_j(t) = w_j \cdot (t-r_j)$, where $r_j$ is the release time of $j$.
	\item Minimizing the sum of weighted $p$th power of flow-times of jobs. This is captured by the function $g_j(t) = w_j \cdot (t-r_j)^p$.
	\item Minimizing the sum of weighted tardiness. This is captured by the function $g_j(t) = w_j \cdot \max \{0,(t-d_j)\}$, where $d_j$ is the deadline of $j$.
\end{itemize} 

In this paper, we design approximation algorithms for minimizing the sum of general delay costs of jobs for the concurrent open-shop scheduling and the precedence constrained scheduling problems.

\subsection{Concurrent Open-shop Scheduling Problem (\os)}
In \os, we are given a set of $m$ machines and a set of $n$ jobs.  Each job $j$ has a release time $r_j$. The main feature of \os is that a job consists of $m$ operations $O_{1j}, O_{2j}, \ldots O_{mj}$, one for each machine $i \in [m]$. 
Each operation $O_{ij}$ needs $p_{ij}$ units of processing on machine $i$. 
We allow operations to have {\em zero processing lengths}.
Throughout the paper, we assume without loss of generality that all our input parameters are positive integers.
A job is complete only when {\em all its operations are complete}. That  is, if $C_j$ denotes the completion time of $j$, then $p_{ij}$ units of operation $O_{ij}$ must be processed in the interval $[r_j, C_j]$ on machine $i$.  In the concurrent open-shop scheduling model, multiple operations of the same job {\em can be processed simultaneously} across different machines. 

The problem has a long history due to its applications in manufacturing, automobile and airplane maintenance and repair \cite{yang1998scheduling}, etc., and has been extensively studied both in operations research and approximation algorithm communities \cite{ahmadi2005coordinated,chen2007supply,garg2007order, mastrolilli2010minimizing, wang2007customer, wagneur1993openshops, leung2007scheduling, bansal2010inapproximability}. 
As minimizing makespan in \os model is trivial, 
much of the research has focused on the objective of minimizing the total weighted completion times of jobs.  The problem was first considered by Ahmadi and Bagchi \cite{ahmadi2005coordinated}, who showed the NP-hardness of the problem. Later, several groups of authors, Chen and Hall \cite{chen2007supply}, Garg et al.\ \cite{garg2007order}, Leung et al.\  \cite{leung2007scheduling}, and  Mastrolilli et al.\ \cite{mastrolilli2010minimizing}, designed 2-approximation algorithms for the problem. Under Unique Games Conjecture, Bansal and Khot showed that this approximation factor cannot be improved for the problem \cite{bansal2010inapproximability}.

Garg, Kumar, and Pandit \cite{garg2007order} studied the more difficult objective of minimizing the total flow-time of jobs, and showed that the problem cannot be approximated better than $\Omega(\log m)$, by giving a reduction from the set cover problem. However, they did not give any approximation algorithm to the problem, and left it as an open problem. To the best of our knowledge, the problem has remained open ever since. In this paper, we make progress on this problem. 

Let $P$ denote the ratio of the maximum processing length among all the operations to the minimum non-zero processing length of among all the operations; that is, $P := \frac{\max_{i,j}\{p_{ij} \}}{\min_{i,j} \{p_{ij}: p_{ij} \neq 0\}}$.


\begin{theorem}
\label{thm:OSUB}
For  the objective of minimizing the sum of general delay cost functions of jobs in the concurrent open-shop scheduling model, there exists a polynomial time $O(\log (m \log P))$ approximation algorithm.
\end{theorem}

We obtain the above result by generalizing the algorithm in Bansal and Pruhs \cite{BansalP10}.
Note that when $m = 1$, our result gives a $O(\log \log P)$ approximation algorithm to the problem, matching the best known polynomial time result in \cite{BansalP10}. 
Recently, for the special case of total weighted flow-time, a constant factor approximation algorithm was obtained by Batra, Garg, and Kumar \cite{batra} when $m = 1$. 
However,  running time of their algorithm is pseudo-polynomial. 
Since the approach in \cite{batra} is very different from the one in \cite{BansalP10}, our result does not generalize \cite{batra}.

As we discussed earlier, the general delay cost functions capture several widely studied performance metrics. Thus, we get:

\begin{corollary}
\label{thm:open-shoparbitrarycost}
There is a polynomial time $O(\log (m \log P))$ approximation algorithm in the concurrent open-shop scheduling model for the following objective functions:
1) Minimizing the sum of weighted flow-times of jobs;
2) Minimizing the weighted $\ell_k$-norms of flow-times of jobs; the approximation factor becomes  $O((\log (m \log P)^{\frac{1}{k}})$);
3) Minimizing the sum of weighted tardiness of jobs.

\end{corollary}

We give the proof Theorem~\ref{thm:OSUB} in Section \ref{sec:openshop}.

\subsection{Precedence Constrained Scheduling Problem (\pc)}


More complex forms of job structures are captured by the precedence constrained scheduling problem (\pc), another problem that has a long history dating back to the seminal work of Graham \cite{graham1966bounds}. Here, we have a set of $m$ identical machines and a set of $n$ jobs; each job $j$ has a processing length $p_j > 0$, a release time $r_j > 0$.  Each job $j$ must be scheduled on exactly one of the $m$ machines. The important feature of the problem is that they are precedence constraints between jobs that capture the computational dependencies across jobs. The precedence constraints are  given by a partial order ``$\prec$'', where a constraint
$j \prec j'$ requires that job $j'$ can only start after job $j$ is completed. Our goal is to schedule (preemptively) each job on exactly one machine to minimize $\sum_{j}g_j(C_j)$.

Precedence constrained scheduling on identical machines to minimize the makespan objective is perhaps the most well-known problem in scheduling theory. Already in 1966, Graham showed that list scheduling gives a 2-approximation algorithm to the problem. Since then several attempts have been made to improve the approximation factor \cite{lam1977worst, gangal2008precedence}. However, Svensson \cite{svensson2010conditional} showed that problem does not admit a $2-\epsilon$ approximation under a strong version of the Unique Games Conjecture introduced by Bansal and Khot \cite{bansal2010inapproximability}. An unconditional hardness of $(4/3-\epsilon)$ is also  known due to Lenstra and Rinnooy Kan \cite{LR78}. Recently, Levey and Rothvoss \cite{levey20161} showed that it is possible to overcome these lowerbounds for the case when $m$ is fixed. An LP-hierarchy lift of the time-index LP with a slightly super poly-logarithmic number of rounds provides a $(1+\epsilon)$ approximation to the problem.

Another problem that is extensively studied in the precedence constrained scheduling model is the problem of minimizing the total weighted completion times of jobs. Note that this problem strictly generalizes the makespan problem, hence all the lowerbounds also extend to this problem. The current best approximation factor of 3.387 is achieved by a very recent result of Li \cite{Li17}. The work builds on a $4$-approximation algorithm due to Munier, Queyranne and Schulz (\cite{MQS98}, \cite{QS06}). 

In a recent work, Agrawal {\em et al} \cite{AgrawalLLM16} initiated the study of minimizing the total flow-time objective in DAG (Directed Acyclic Graphs) {\em parallelizability} model. In this model, each job is a DAG, and a job completes only when all the nodes in the DAG are completed. For this problem,  they showed greedy online algorithms that are constant competitive when given $(1+\epsilon)$-speed augmentation.  
The DAG parallelizability model is a special case of \pc.  However, as there are no dependencies between jobs, and individual nodes of the DAG do not contribute to the total flow-time unlike in \pc, complexity of the problem is significantly different from \pc. For example, when there is only one machine, the DAG structure of individual jobs does not change the cost of the optimal solution, and hence the problem reduces to the standard single machine scheduling problem. Therefore, scheduling jobs using Shortest Remaining Processing Time (SRPT), where processing length of a DAG is its total work across all its nodes, is an optimal algorithm. (Within a DAG, the nodes can be processed in any order respecting the precedence constraints.)

On the other hand, we show a somewhat surprising result for the \pc problem.
We show that the problem of minimizing the total flow-times of jobs does not admit any reasonable approximation factor even on a {\em single machine}. This is in sharp contrast to makespan and the sum of weighted completion times objective functions that admit $O(1)$-approximation algorithms even on multiple machines. Our hardness proof is based on a recent breakthrough work of Manurangsi \cite{Manurangsi17} on approximating the Densest-k-Subgraph (\dks) problem.

\begin{theorem}
	\label{thm:PCLB}
	In the precedence constrained scheduling model, for  the objective of minimizing the total flow-times of jobs on a single machine, no polynomial time algorithm can achieve an approximation factor better than $n^{\frac{1}{(\log \log n)^c}}$, for some universal constant $c >0$, assuming the exponential time hypothesis (ETH). 
\end{theorem}

To circumvent this hardness result, we study the problem in the speed augmentation model, which can also be thought of as a bi-criteria analysis. In the speed augmentation model, each machine is given some small extra speed compared to the optimal solution.  The speed augmentation model was introduced in the seminal work of Kalyanasundaram and Pruhs \cite{kirk} to analyze the effectiveness of various scheduling heuristics in the context of online algorithms. However, the model has also been used in the offline setting to overcome strong lowerbounds on the approximability of various scheduling problems; see for example results on non-preemptive flow-time scheduling that use $O(1)$-speed augmentation to obtain $O(1)$-approximation ratio  \cite{bansal2007non, im2015dynamic}.

Our second main result is an $O(1)$-approximation algorithm for the problem in the speed augmentation model. Previously, no results were known for the flow-time related objective functions for \pc.

\begin{theorem}
\label{thm:PCUB}
For  the objective of minimizing the sum of general delay cost of jobs in the precedence constrained scheduling model on identical machines, there exists a polynomial time $O(1)$-speed $O(1)$ approximation algorithm. Furthermore, the speed augmentation required to achieve an approximation factor better than $n^{1-c}$, for any $c > 0$, has to be at least the best approximation factor of the makespan minimization problem. The lowerbound on speed augmentation extends to any machine environment, such as related and unrelated machine environments. 
\end{theorem}

\medskip
We give the proofs of above theorems in Section \ref{sec:mig}.

\subsection{Applications of \os and \pc in Data Center Scheduling}
\label{apps}
Besides being fundamental optimization problems, \os and \pc models are very closely related to the scheduling problems that arise in the context of data centers. In particular, \os is a special case of the {\em Coflow Scheduling} problem  introduced in a very influential work of Choudary and Stoica \cite{chowdhury2012coflow, chowdhury2014efficient}. On the other hand, \pc generalizes the DAG scheduling problem, again a widely studied problem in systems literature. In fact, the DAG scheduling model has been adopted by Yarn, the resource manager of Hadoop \cite{hadoop}. See  \cite{GrandlKRAK16,grandl2016altruistic} and references there-in for more details.  

Besides being fundamental optimization problems, \os and \pc models are very closely related to the scheduling problems that arise in the context of data centers.  We briefly describe the relevance of \os and \pc to scheduling in data centers. 

\paragraph{Coflow Scheduling} The \os problem is a special case of the {\em coflow scheduling abstraction}  introduced in a very influential work of Choudary and Stoica \cite{chowdhury2012coflow, chowdhury2014efficient}, in the context of scheduling data flows. They defined coflow as a collection of parallel flows with a common performance goal. Their main motivation to introduce coflow  was that in big-data systems, MapReduce systems for example, communication is structured and takes place across machines in successive computational stages. In most cases, communication stage of jobs cannot finish until all its flows have completed. For example, in MapReduce jobs, a reduce task cannot begin until {\em all} the map tasks finish. Although coflow abstraction was introduced to model scheduling flows in big data networks, it can also be applied to job scheduling in clusters; see \cite{grandl2016altruistic} for example.

Therefore, we describe a slightly more general version of the coflow abstraction.  Here, we are given a set of $m$ machines (or $m$ resources). Each job $j$ consists of a set of operations $\{O_{ij}\}$ for $i =1,2...$. Associated with each operation $O_{ij}$ is a demand vector $D_{ij} = (d_{ij}(1), d_{ij}(2), ...,d_{ij}(m))$, where each $d_{ij}(k)\in \{0,1\}$. The demand vector indicates the subset of machines or resources the operation requires. An operation can be executed only when all the machines in the demand vector are allocated to it. Moreover, for each operation we are also given a processing length $p_{ij}$. The goal is to schedule all operations such that at any time instant the capacity constraints on machines are not violated: that is, each machine is allocated to exactly one operation. A job completes only when all its operations have finished.   

The coflow problem studied by Choudary and Stoica \cite{chowdhury2012coflow, chowdhury2014efficient} corresponds to the case where  the demand vectors of all jobs have exactly two 1s. That is, every operation needs two machines to execute. The machines typically correspond to input and output ports in a communication link.  
On the other hand, if the demand vector $D_{ij}$ of each operation $O_{ij}$ consists of exactly one non-zero entry, then the coflow scheduling is equivalent to \os.

In the past few years, coflow scheduling has attracted a lot of research both in theory and systems communities. In practice, several heuristics are known to perform well \cite{chowdhury2012coflow, chowdhury2014efficient, chowdhury2015efficient, grandl2016altruistic} for the problem.
The theoretical study of coflow scheduling was initiated by Qiu, Stein, and Zhong \cite{qiu2015minimizing}. By exploiting its connections to \os, they designed a constant factor approximation algorithm to the objective of minimizing the total weighted completion times of jobs. Building on this work, better approximation algorithms were designed in \cite{khuller2016brief, ahmadi2017scheduling, shafiee2017improved}. Unfortunately, the techniques developed in these works do not seem to extend to the flow-time related objectives. 

\paragraph{DAG Scheduling}  Another problem that has attracted a lot of research in practice is the DAG scheduling problem. In this problem, we are given a set of machines (or clusters), and a set of jobs. Each job has a weight that captures the priority of a job. Each job $j$ is represented by a directed acyclic graph (DAG). Each node of a DAG represents a task -- a single unit of computation-- that needs to executed on a {\em single machine}. Each task has a release time and a processing length. An edge $j \rightarrow j'$ in the DAG indicates that the task $j'$ depends on the task $j$, and $j'$ can not begin its execution unless  $j$  finishes.
The goal is to schedule jobs/DAGs on the machines so as to minimize the total weighted flow-time of jobs. Interestingly, this is one of the models of job scheduling that has been adopted by Yarn, the resource manager of Hadoop \cite{hadoop}. (Hadoop is a popular implementation of MapReduce framework.) Because of this, DAG scheduling has been a very active area of research in practice; see \cite{GrandlKRAK16,grandl2016altruistic} and references there-in for more details. 

It is not hard to see that DAG scheduling problem is a special case of the precedence constrained scheduling problem. 
The union of the individual DAGs of jobs can be considered as one DAG, with appropriately defined release times and weights for each node. Furthermore, if jobs have no weight, and the release times of all tasks in the same DAG are equal, then the DAG scheduling model described above is same as the DAG parallelizability model studied by Agrawal {\em et al} \cite{AgrawalLLM16}. In fact, they design $(1+\epsilon)$-speed $O(1/\epsilon)$-competitive {\em online algorithm} for the problem.
On the other hand, our approximation algorithm for \pc gives an approximation algorithm for the DAG scheduling problem even with weights and arbitrary release times for individual tasks.

%% file: ourtechniques.tex
\subsection{Overview of the Algorithms and Techniques}
	\label{sec:ot}

Both our algorithms are based on rounding linear programming relaxations of the problems. However, individual techniques are quite different, and hence we discuss them separately.

\paragraph{Open-shop Scheduling Problem} Our algorithm for  \os is based  on the geometric view of scheduling developed by Bansal and Pruhs \cite{BansalP10} for the problem of minimizing the sum of general delay costs of jobs on a single machine, which is a special case of our problem when $m = 1$. 
The key observation that leads to this geometric view is that minimizing general delay costs is equivalent to coming-up with a set of {\em feasible deadlines} for all jobs. Moreover, testing the feasibility of deadlines further boils down to ensuring that for {\em every} interval of time, the total volume of jobs that have (arrival time, deadline) windows within that interval is not large compared to the length. By a string of nice arguments, the authors show that this deadline scheduling problem can be viewed as a capacitated geometric set cover problem called R2C, which stands for capacitated rectangle covering problem in $2$-dimensional space.\footnote{``C'' stands for the capacitated version in which rectangles have capacities and points have demands. Later, we shall use ``M'' for the multi-cover version, where rectangles are uncapacitated (or have capacity 1) and points have different demands.  We use ``U'' for the uncapacitated version, where rectangles are uncapacitated and all points have demand 1.} Further, they argue that an $\alpha$-approximation algorithm for R2C problem can be used to obtain an $\alpha$-approximation algorithm for the scheduling problem.

In R2C, we are given a set $\pt$ of points in 2 dimensions, where each $p \in \pt$ is specified by coordinates $(x_p, y_p)$. Associated with each point $p \in \pt$ is a demand $d_p > 0$.  We are also given a set $\rs$ of rectangles, where $r \in \rs$ has the form $(0, x_r) \times (y^1_r, y^2_r)$.   Each rectangle $r$ has a capacity $c(r) > 0$ and a cost $w(r) > 0$.  The goal is to choose a minimum-cost set of rectangles, such that for every point $p \in \pt$, the total capacity of selected rectangles covering $p$ is at least $d_p$.  

Our problem, which we call \prc, can be seen as a parallel version of R2C. 
In \prc, we have $m$ instances of R2C problem with a common set of rectangles. Namely, the $i$th instance is defined by $(\pt_i, \rs)$, where each $p \in \pt_i$ is associated with a demand $d_p$, each $r \in \rs$ is associated with a capacity $c(i, r)$ and cost $w(r)$. Notice that a rectangle $r \in \rs$ has the same cost across the $m$ instances, but has different capacities in different instances.   The goal is to find a minimum cost set $X \subseteq \rs$ of rectangles that is a valid solution for every R2C instance $(\pt_i, \rs)$.
Using the arguments similar to \cite{BansalP10}, we also show that if there is an $\alpha$-approximation to \prc problem, it gives an $\alpha$-approximation for the \os problem.

\begin{figure}[!ht]
	\centering
	\includegraphics[width=0.6\textwidth]{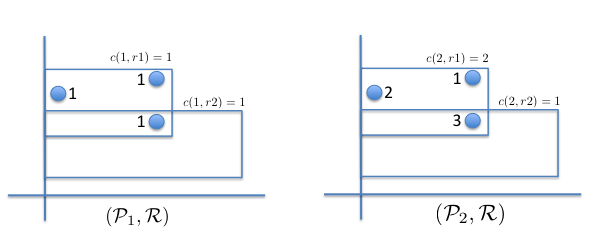} 	
	\caption{\label{fig:pr2c} The figure shows an instance of \prc problem where are there are two sets of points. For this instance the only feasible solution is to pick both the rectangles.}
\end{figure}

Thus, much of our work is about designing a good approximation algorithm for the \prc problem. 
Our algorithm for \prc is a natural generalization of the algorithm for R2C of Bansal and Pruhs in \cite{BansalP10}. 
In \cite{BansalP10}, Bansal and Pruhs formulated an LP relaxation for R2C based on knapsack cover (KC) inequalities.  In the LP, $x_r \in [0, 1]$ indicates the fraction of the rectangle $r$ that is selected.  If the LP solution picks a rectangle to some constant fraction, then we can also select the rectangle in our solution without increasing the cost by too much. After selecting these rectangles, some points in $\pt$ are covered, and some other points $p$ will still have residual demands $d'_p$.  These not-yet-covered points are divided into two categories: {\em light and heavy points}. Roughly speaking, a point $p$ is heavy if it is mostly covered by rectangles $r$ with $c(r) \geq d'_p$  in the LP solution; a point $p$ is light if it is mostly covered by rectangles $r$ with $c(r) < d'_p$. Heavy points and light points are handled separately. The problem of covering heavy points will be reduced to the R3U problem, a geometric weighted set-cover problem where elements are points in 3D space and sets are axis-parallel cuboids. On the other hand, the problem of covering light points can be reduced to $O(\log_2P)$ R2M instances. Each R2M instance is a geometric weighted set multi-cover  
instance in 2D-plane.  
By appealing to the geometry of the objects produced by the scheduling instance,  \cite{BansalP10} prove that {\em union complexity of objects} in R3U and R2M instances is small. In particular, R3U instance has \emph{union complexity} $O(\log P)$ and each R2M instance has union complexity $O(1)$.
Using the technique of quasi-uniform sampling introduced in \cite{varadarajan2009epsilon, chan2012weighted} for solving geometric weighted set cover instances with small union complexity, Bansal and Pruhs obtain $O(\log \log P)$ and $O(1)$ approximation ratios for the problems of covering heavy and light points, respectively. 

\medskip	
In our problem, we have $m$ parallel R2C instances with a common set of rectangles.  As in \cite{BansalP10}, for each instance,  we categorize the points into heavy and light points based on the LP solution.  The problem of handling heavy points then can be reduced to $m$ R3U instances, with the $m$ sets of cuboids identified.  However, we cannot solve these $m$ R3U instances separately, because such a solution cannot be mapped back to a valid schedule for \os. Therefore, we combine the $m$ instances of R3C into a single instance of a 4 dimensional problem.
So, in the combined instance, our geometric objects, which we call hyper-4cuboids, contain (at most) one 4-cuboid (a cuboid in 4 dimensions) from each of the $m$ instances.
The goal is to choose a minimum cost set of objects to cover all the points.   
On the other hand, for the light points, \cite{BansalP10} reduced the problem to $O(\log P)$ R2M instances. Again, this is approach is not viable for our case as we need to solve all the instances in parallel. By a simple trick, we first merge the $O(\log P)$ instances into one R2M instance.  We then have $m$ R2M instances with a new sets of rectangles identified, which we map into a single 3-dimensional geometric multi-set cover problem. 

In both cases we show that the union complexity of the objects in our geometric problems increase at most by a factor of $m$ compared to the objects in \cite{BansalP10}.  Thus, we have $O(m \log P)$ union complexity for the problem for heavy points and $O(m)$ union complexity for the problem for light points. Using the technique of \cite{BansalP10}, we obtain $O(\log (m \log P))$ and $O(\log m)$ approximation ratios for heavy and light points respectively, resulting in an $O(\log (m \log P))$ overall approximation ratio.

\medskip
\paragraph{Precedence Constrained Scheduling} Our algorithm for the precedence constrained scheduling problem works in two steps. In the first step, we construct a \emph{migratory} schedule, in which  a job may be processed on multiple machines.  For migratory schedules, we can assume that all jobs have unit size by replacing each job $j$ of size $p_j$ with a precedence chain of $p_j$ unit length jobs. Solving a natural LP relaxation for the problem gives us a completion time vector $(C_j)_j$.  Then we run the list-scheduling algorithm of Munier, Queyranne and Schulz \cite{MQS98}, and Queyranne and Schulz \cite{QS06}, that works for the problem with weighted completion time objective.  Specifically, for each job $j$ in non-decreasing order of $C_j$ values, we insert $j$ to the earliest available slot after $r_j$ without violating the precedence constraints.  

To analyze the completion time of job $j$, we focus on the schedule constructed by our algorithm after the insertion of $j$.  A simple lemma is that at any time slot after $r_j$, we are making progress towards scheduling $j$ in the schedule: either all machines are busy in the time slot, or we are processing a set of jobs whose removal will decrease the ``depth'' of $j$ in the precedence graph.  This lemma was used in \cite{MQS98, QS06} to give their $3$-approximation for the problem of minimizing weighted completion time (for unit-size jobs), and recently by Li \cite{Li17} to give an improved $2.415$-approximation for the same problem.  With $3$-speed augmentation, this leads to a schedule that completes every job $j$ by the time $C_j$. With additional speed augmentation, this leads to an $O(1)$-approximation for the problem with general delay cost functions. In the second step,  we convert the migratory schedule into a non-migratory one, using some known techniques (e.g. \cite{kalyanasundaram2001eliminating}, \cite{Kulkarni17}, \cite{ImM16}).  The conversion does not increase the completion times of jobs,  but requires some extra $O(1)$-speed augmentation.

%% file: openshop.tex
\section{Concurrent Open-shop Scheduling}
	\label{sec:openshop}

In this section we consider the concurrent open-shop scheduling problem. Recall that in \os, we are given a set of $m$ machines and a set of $n$ jobs.  Each job has a release time $r_j$. A job consists of $m$ operations $\{O_{ij}\}_i$, one for each machine $i \in [m]$. Each operation $O_{ij}$ needs $p_{ij}$ units of processing on machine $i$. A job finishes only when all its operations are completed. The goal is to construct a preemptive schedule that minimizes the sum of costs incurred by jobs: $\sum_j g_j(C_j)$. 

As we mentioned earlier, our algorithm for \os is based on the geometric view of scheduling developed in the work of Bansal and Pruhs \cite{BansalP10}. Similar to \cite{BansalP10}, we first reduce our problem to a geometric covering problem that we call Parallel Capacitated Rectangle Covering Problem (\prc). We argue that an $\alpha$-approximation to \prc will give an $\alpha$ approximation to our problem. Then, we design a $O(\log (m\log P))$-factor approximation algorithm to \prc.

\subsection{Reduction to \prc problem.}

In \prc, we have $m$ instances of R2C problem with a common set of rectangles. Input to the problem consists of sets $\pt_1, \pt_2, \ldots \pt_m $, where each $\pt_i$ is a set of points in 2-dimensional space. Each point $p \in \pt_i$ is specified by its coordinates $(x_p, y_p)$ and has a demand $d_p > 0$.  The input also consists of a set $\rs$ of rectangles. Each rectangle $r \in \rs$ has the form $(0, x_r) \times (y^1_r, y^2_r)$, and has a cost $w(r)$. Each rectangle $r$ also has a capacity $c(i,r)$ which depends on the point set $\pt_i$. Notice that a rectangle $r \in \rs$ has the same cost across the $m$ instances, but has different capacities in different instances.   The goal is to find a minimum cost set $X \subseteq \rs$ of rectangles that is a valid solution for every R2C instance $(\pt_i, \rs)$. Recall that a set of rectangles is a valid solution to an instance of R2C, if for every point the total capacity of rectangles that cover the point is at least the demand of the point.

We can capture the \prc problem using an integer program. Let $x_r$ be a binary variable that indicates if the rectangle $r \in \rs$ is picked in the solution or not. Then, the following integer program IP (\ref{IP:prc} - \ref{e:prcnonnegativity}) captures the \prc problem.
\vspace{-3mm}
\begin{align}
\textstyle  \text{Minimize} \quad \sum_{ r \in \rs} w(r)x_r   \label{IP:prc}  \vspace{-3mm} 
\end{align}
\vspace{-10mm}
\begin{align}
\textstyle  \forall i, \forall p \in \pt_i &:  &\hspace{2mm} \sum_{r \ni p}  c(i,r) x_r  &\geq d_p \label{e:prccover} \\
\textstyle \forall r \in \rs&: &\hspace{2mm} x_r \in \{0, 1\} \label{e:prcnonnegativity} 
\end{align}

To see the connection between \os and \prc, we need to understand the structure of feasible solutions to \os. Consider a feasible schedule $S$ to an instance of \os. Suppose the completion time of a job $j$ is $C_j$ in $S$.  This implies that on every machine $i$,  the job $j$ completes its operation $O_{ij}$ in the interval $[r_j, C_j]$. If $S$ is a feasible schedule and $\{C_j\}_j$ is the completion times of jobs, then processing jobs using the Earliest Deadline First (\edf) algorithm on each machine ensures that all jobs finish by their deadline $C_j$. Thus, one of the main observations behind our reduction is that minimizing the sum of costs incurred by jobs is equivalent to coming up with a deadline $C_j$ for each job $j$. Thus a natural approach is to formulate  \os as a deadline scheduling problem. However, this raises the question: Is there a way to test if a set of deadlines $\{C_j\}_j$ is a feasible solution to \os? An answer to the question is given by the characterization of when \edf will find a feasible schedule (one where all jobs meet their deadlines) on a single machine. To proceed, we need to set up some notation.

For an interval $I = [t_1, t_2]$, which consists of all the time slots between $t_1$ and $t_2$, let $J(I)$ denote the set of jobs that have release times in $I$: $J(I) = \{j | r_j \in [t_1, t_2] \}$. For a set of jobs $J'$ and a machine $i$, let $P(i, J')$ denote the total processing length of operations $O_{ij}$ of the jobs in the set $J'$. Now, we introduce the notion of {\em excess} for an interval. 

\begin{definition}
For a given interval $I$ and a machine $i$, the excess of $I$ on machine $i$, denoted by $\xi(i, I)$, is defined as $\max \{0, P(i, J(I)) - |I| \}$.
\end{definition}

Following lemma states the condition under which \edf can schedule all the jobs within their deadlines.

\begin{lemma}
\label{lem:edfch}
Given a set of jobs with release times and deadlines, scheduling operations on each machine according to \edf is feasible if and only if, for every machine $i$ and every interval $I = [t_1, t_2]$, the total processing lengths of operations corresponding to jobs in $J(I)$ that have deadlines greater than $t_2$ is at least $\xi(i, I)$.
\end{lemma}

Clearly, if the total processing lengths of operations corresponding to jobs in $J(I)$ that have deadlines greater than $t_2$ is less  than $\xi(i, I)$, then no algorithm can meet all the deadlines, as the length of operations scheduled in the interval $I$ is greater than $|I|$. Sufficiency of the above lemma follows from a bipartite graph matching argument, and we refer the reader to \cite{BansalP10} for more details.

Lemma \ref{lem:edfch} leads to an integer programming formulation for \os, and we shall use this IP to define the \prc instance. Let $L = \max_{i} \{ \sum_{j} p_{ij} \}$; note that every reasonable schedule finishes by $L$.  For every job $j$, and every integer $q \in \left[-1, \ceil{\log_2 g_j(L)}\right]$, let $t_{j, q}$ be the largest integer $t  \in (r_j, L]$ such that $g_j(t) \leq 2^q$ (if such $t$ does not exist, then $t_{j, q} = r_j$). Let $t_{j, -2} = r_j$. For every $j$ and $q\in \left[-1, \ceil{\log_2 g_j(L)}\right]$, we have a variable $x_{j,q}$ indicating whether $C_j \in \left(t_{j, q-1}, t_{j, q}\right]$. Therefore, for a job $j$, the total number of $x_{j,q}$ variables in our IP will be at most $O(\log g_j(L))$.   For every $j$ and every $t > r_j$, let $q_{j, t}$ be the integer $q$ such that $t \in (t_{j, q-1}, t_{j, q}]$.

Our IP for \os is as follows:

\vspace{-4mm}
\begin{align}
\textstyle  \text{Minimize} \quad \sum_{j} \sum_{q} \floor{2^q} x_{j,q}   \label{IP:OS}  \vspace{-3mm} 
\end{align}
\vspace{-7mm}
\begin{align}
\textstyle \forall i, \forall I = [t_1, t_2]&: &\hspace{2mm} \sum_{j \in J(I)} p_{ij} x_{j,q_{j, t_2}} &\geq \xi(i, I) \label{e:OSEdf} \\
\textstyle \forall j, \forall q &: &\hspace{2mm} x_{j,q} &\in \{ 0,1 \} \label{e:OSnn}
\end{align}

We shall argue that the value of above IP is within a factor of $O(1)$ times the optimum cost of the scheduling problem.  First, given a feasible schedule $S$ for \os with completion time vector $(C_j)_j$, we construct a solution to the above integer program with a small cost.  For each $j$ and integer $q$ such that $2^{q-1} \leq g_j(C_j)$, we let $x_{j, q} = 1$; set all other $x$ variables to $0$. Clearly, the cost of the solution to the IP is at most $O(1)$ times the cost of $S$.  Moreover, Constraint~\eqref{e:OSEdf} is satisfied: we have that $\sum_{j \in J(I):C_j > t_2}p_{i, j} \geq \xi(i, I)$. If $C_j > t_2$ for some $j \in J(I)$, then $x_{j, q_{j, t_2}} = 1$, as $2^{q_{j,t_2}-1} \leq g_j(t_2) \leq g_j(C_j)$.

On the other hand, if we are given an optimum solution to the above IP, we can convert it into a schedule for \os of cost at most the cost of the IP. For every $j$, let the completion time $C_j$ of the job to  be $t_{j, q}$, where $q$ is the largest number such that $x_{j, q} = 1$; such a $q$ must exist in order to satisfy the Constraint 
(\ref{e:OSEdf}) for the interval $t_1 = t_2 = r_j$.  Then the cost of our schedule is at most the cost of IP.  On the other hand, Constraint~\eqref{e:OSEdf} says that $\sum_{j \in J(I)} p_{ij} x_{j,q_{j, t_2}} \geq \xi(i, I)$, implying $\sum_{j \in J(I), C_j > t_2} p_{ij} \geq \xi(i, I)$, as $x_{j, q_{j,t_2}} = 1$ implies $C_j \geq t_2$.

\begin{remark}
We will not discuss the representation of arbitrary delay functions $g_j(t)$ and how to compute $t_{j, q}$ here. We refer the readers to \cite{BansalP10} for more details. Running time of all our algorithms will be polynomial in $n, m, \log L$ and $\max_{j}\{\log g_j(L)\})$.
\end{remark}

The above integer program hints towards how we can interpret \os geometrically as  \prc. Indeed, it is equivalent to \prc problem.  For every machine $i$ in \os, we create a set $\pt_i$ in \prc. For every interval $I = [t_1, t_2]$ with $\xi(i, I) > 0$, we associate a point $p_I = (t_1, t_2)$ in 2-dimensional space.  The demand $d_{p_I}$ of a point $p_I \in \pt_i$ is equal to $\xi(i, I)$, the excess of interval $I$ on machine $i$. This completes the description of the point sets $\pt_1, \pt_2, \ldots \pt_m$ in \prc. Now we define the rectangle set $\rs$. We shall create a rectangle for each job $j$ and each $q$ for which $x_{j, q}$ is defined. Notice that $x_{j, q}$ appears on the left side of Constraint~\eqref{e:OSEdf} if $r_j \in [t_1, t_2]$ and $t_2 \in (t_{j, q-1}, t_{j, q}]$. Thus, we shall let the rectangle for the variable $x_{j, q}$ to be $(0, r_j] \times (t_{j, q-1}, t_{j, q}]$. This rectangle has cost $2^q$; for the instance $\pt_i$, it has a capacity $p_{ij}$. Notice that the rectangle for $x_{j, q}$ covers a point $(t_1, t_2)$ if and only if $t_1 \leq r_j \leq t_{j, q-1} < t_2 \leq t_{j, q}$; in other words, the job $j$ is released in the interval $[t_1, t_2]$, and its completion time is greater than $t_2$, 
which is exactly what we want. Thus, the IP (\ref{IP:OS} - \ref{e:OSnn}) is equivalent to our \prc problem. We now forget about the \os problem and focus exclusively on designing a good algorithm for the \prc problem. 

\subsection{Algorithm For \prc Problem}
Our algorithm for \prc is based on rounding a linear programming relaxation of the IP (\ref{IP:prc} - \ref{e:prcnonnegativity}). As pointed out in \cite{BansalP10}, the linear programing relaxation of the IP (\ref{IP:prc} - \ref{e:prcnonnegativity}) obtained by relaxing the variables $x_r \in [0,1]$ has a large integrality gap even when  there is only one set of points. We strengthen our LP by adding the so called KC inequalities, introduced first in the work of Carr {\em et al} (\cite{carr2000strengthening}). Towards that we need to define $c(i, S)$, which indicates the total capacity of rectangles in $S \subseteq \mathcal{R}$ with respect to a point set $\mathcal{P}_i$: $c(i,S) = \sum_{r \in S} c(i,r)$. We are ready to write the LP.
\begin{equation}
 \text{Minimize} \quad \sum_{r \in \mathcal{R}} w(r)x_r   \label{lp:prc} 
\end{equation}\vspace*{-20pt}
\begin{align}
\forall i, \forall p \in P_i, \forall S \subseteq \mathcal{R} &: & \sum_{r \in \mathcal{R} \setminus S: p \in r} \left( \min \{ c(i,r), \max \{0, d_p - c(i, S)\} \} \right) \cdot  x_r   
&\geq d_p - c(i, S) \label{e:lpprccover} \\ 
\textstyle \forall r \in \mathcal{R}&: &\hspace{2mm} x_r &\geq 0\label{e:lpprcnonnegativity} 
\end{align}

Let us focus on the  KC inequalities Eq.(\ref{e:lpprccover}) for a point set $\pt_i$. Fix a point $p \in \pt_i$ and a set of rectangles $S \in \mathcal{R}$. Recall that $p$ has a demand of $d_p$. Suppose, in an integral solution all the rectangles in $S$ are chosen. Then, they contribute at most $c(i,S)$ towards satisfying the demand of the point $p$. The remaining rectangles still need to account for $d_p - c(i, S)$. Notice also that in the KC constraints we truncate the capacity of a rectangle $r$ to $\min \{ c(i,r), \max \{0, d_p - c(i, S)\} \}$. This ensures that the LP solution does not cheat by picking a  rectangle with a large capacity to a tiny fraction.  Clearly, the truncation has no effect on the integral solution.

There are exponentially many KC constraints in our LP. However, using standard arguments, we can solve the LP in polynomial time to get a $(1+\epsilon)$-approximate solution, for any $\epsilon > 0$, which suffices for our purposes to obtain a logarithmic approximation to \os; see \cite{BansalP10, carr2000strengthening}  for more details. The rest of this section is devoted to rounding the LP solution.

\medskip
\noindent \textbf{Weighted Geometric Set Multi-Cover Problem.}
The main tool used in our rounding algorithm is the result by \cite{bansal2012weighted} for the {\em weighted geometric set multi-cover problem}. Similar to the standard set cover problem, in this problem, we are given a set $U$ of points and a set $M$ consisting of subsets of $U$; typically, the sets in $M$ are defined by geometric objects. Further, each point $p \in U$ has a demand $d_p$, and a set $ r \in M$ has a weight $w(r).$ The goal is to find  the minimum weight set $S \subseteq M$ such that every point $p$ is covered by at least $d_p$ number of subsets in $S$. Note crucially that the sets in $M$ do not have capacities. If the sets have capacities, then the problem becomes similar to \prc.  
The interesting aspect of geometric set cover problems is that sets in $M$ are defined by geometric objects, hence subsets in $M$ have more structure than in the standard set cover problem. In particular, if the sets in $M$ have low union complexity, then the problem admits a better than $O(\log |M|)$ approximation algorithm. Now we introduce the concept of union complexity of sets.

\begin{definition}
Given a set $S$ of $n$ geometric objects, the union complexity of $S$ is the number of edges in the arrangement of the boundary of $S$.
\end{definition}

We will not get into the details of the definition, and we refer the readers interested in knowing more to \cite{matouvsek2002lectures} for an excellent introduction or to  \cite{varadarajan2009epsilon, chan2012weighted, bansal2012weighted}.
For our purposes, it suffices to know that the union complexity of 3-dimensional objects is the {\em total number of vertices, edges, and faces on the boundary}.  It turns out that the geometric objects in our rounding algorithm are 4-dimensional objects. However, by appropriate projection,  we reduce our problem of bounding the union complexity of  4-dimensional objects to that of 3-dimensional ones.

For geometric objects with low union complexity, building on the breakthrough works of Chan {\em et al}  \cite{chan2012weighted} and Varadarajan \cite{varadarajan2009epsilon}, Bansal and Pruhs \cite{bansal2012weighted} showed the following theorem.

\begin{theorem}
\label{thm:UC}
Suppose the union complexity of every $S \subseteq M$ of size $k$ is at most $kh(k)$. Then, there is a polynomial time $O(\log h(n))$ approximation algorithm for the weighted geometric set multi-cover problem. Further, the approximation factor holds even against a feasible fractional solution.
\end{theorem}

\paragraph{Rounding} Our rounding algorithm is based on reducing our problem to several  instances of the geometric set cover problem, and then appealing to Theorem \ref{thm:UC}. Consider an optimal solution $\vec{x} = \{x_r\}_r$ to the LP (\ref{lp:prc} - \ref{e:lpprcnonnegativity}).  Let $\beta > 1/20$ be some constant. We first scale the  solution by a factor of $1/\beta$, and let $\vec{x'} = \{ x'_r\}_r$ be this scaled solution. Clearly, scaling increases the cost of LP solution at most by a factor of $1/\beta$. Our rounding consists of three steps. 

In the first step, we pick all rectangles $r$ for which $x'_r \geq 1$. Let $\lps$ denote this set. Let $\mathcal{P}_i(\lps)$ denote the set of points that are covered by rectangles in $\lps$. We modify the point sets $\mathcal{P}_i$ for $i = 1,2, ..., m$, by removing the points that are already covered by $\lps$. For all $i \in [m]$, let $\mathcal{P}_i' = \mathcal{P}_i \setminus S$. In the second step, for each $\mathcal{P}'_i$ we classify the points in the set into two types, {\em heavy} or {\em light}, based on the LP solution. For heavy points, we create an instance of the geometric set cover problem, and then appeal to Theorem \ref{thm:UC}. The main technical difficulty here is to bound the union complexity of the geometric objects in our instance. For the light points, we create another separate instance of the geometric set multi-cover problem, and then apply the Theorem \ref{thm:UC}. Finally, we obtain our solution by taking the union of the rectangles picked in all three steps.

\medskip
Fix a set $\mathcal{P}'_i$ of uncovered points. For a point $p \in \mathcal{P'}_i$, let $\lps(p)$ denote the set of rectangles in $\lps$ that contain $p$. That is, $\lps(p) = \lps \cap \{ r: r \in \mathcal{R}, p \in r \}$. For a point $p$, define the residual demand of $p$ as $d_p - c(i, \lps_p)$. From the definition of set $\mathcal{P}'_i$, for every point $p \in \mathcal{P}'_i$, we note that $d_p - c(i, \lps_p) > 0$. We now apply the KC inequalities  on the set $\lps(p)$ for each point $p$.  From Eq.(\ref{e:lpprccover}) we have, for all $p \in \mathcal{P}'_i:$
\begin{eqnarray*}
 \quad \sum_{r \in \mathcal{R}\setminus \lps(p): p \in r} \left(\min \{c(i,r), d_p - c(i, \lps(p)) \} \right) \cdot x_r 
\geq d_p - c(i, \lps(p)) 
\end{eqnarray*}
which implies that our scaled solution $\vec{x'}$ satisfies  for all $p \in \mathcal{P}'_i$:
\begin{eqnarray*}
 \sum_{r \in \mathcal{R}\setminus \lps(p): p \in r} \left(\min \{c(i,r), d_p - c(i, \lps(p)) \} \right) \cdot x'_r 
 \geq \frac{d_p - c(i, \lps(p))}{\beta}.
 \end{eqnarray*} 

Note that for all $r$,  $x'_r \in [0,1]$; otherwise, we would have picked those rectangles in $\lps$. Next we round the residual demands of points and the capacities of rectangles as follows: For a point $p \in \mathcal{P}'_i$, let $\tilde{d}_p$ denote the residual demand of $p$ rounded up to the nearest power of 2. On the other hand, we round down the capacities $c(i,r)$ of rectangles $r \in \mathcal {R}$ to the nearest power of 2; let $\tilde{c}(i,r)$ denote this new rounded down capacities. Since $\vec{x'}$ is scaled by a factor of $1/\beta$ we still have that for all $p \in \mathcal{P}'_i$

\begin{equation}
\label{e:scaledcap}
 \sum_{r \in \mathcal {R} \setminus \lps(p): p \in r} \left(\min \{ \tc(i,r), \td_p  \} \right) \cdot x'_r \geq \frac{\td_p}{4 \beta} \geq 3 \td_p . 
\end{equation}

To classify the points into heavy and light, we need the notion of class.

\begin{definition}
Let $\tc_{\min}(i) = \min_{r \in \mathcal{R}} \{ \tc(i,r) \}$.  A rectangle $r \in \mathcal{R}$ is a class $k$ rectangle with respect to a point set $\mathcal{P}_i$ if $\tc(i,r) = 2^k \cdot \tc_{\min}(i)$. We say that a point $p \in \mathcal{P'}_i$ is a class $k$ point if $\td_p = 2^k \cdot \tc_{\min}(i)$.
\end{definition}

Recall that the capacities of rectangles for different point sets can be different. Therefore, the class of a rectangle depends on the point set $\mathcal{P'}_i$.
Now we categorize points in $\mathcal{P'}_i$ into heavy and light as follows.

\begin{definition}
A point $ p \in \mathcal{P'}_i$ belonging to class $k$ is heavy if its demand is satisfied by the rectangles of class at least $k$ in the LP solution. Otherwise, we say that the point is light.
\end{definition}

From the definition, for all heavy points $p \in \pt'_i$ we have
\begin{equation}
\label{eq:heavypoint}
\sum_{r \in \rs \setminus \lps: \tc(i,r) \geq \td_p, p \in r} x'_r \geq 1,
\end{equation}
and from Eq.(\ref{e:scaledcap}) all light points $p \in \pt'_i$ satisfy
\begin{equation}
\label{eq:lightpoint}
\sum_{r \in \rs \setminus \lps: \tc(i,r) < \td_p, p \in r} \tc(i,r) x'_r \geq (1/4\beta - 1) \td_p =  2\td_p.
\end{equation}

Let $\pt^1_i \subseteq \pt'_i$ denote the set of heavy points and let $\pt^2_i \subseteq \pt'_i$ denote the set of light points. Let $ \mathcal{R'} =  \mathcal{R} \setminus \lps$.
We create two separate instances, a heavy instance and a light instance of the $\prc$ problem, corresponding to the set of heavy points and the set of light points. The capacities of rectangles in these instances will be their rounded down values: that is, the capacity of a rectangle $r \in \mathcal{R'}$ for the point sets $\pt^*_i$ will be $\tc(i,r)$.  Similarly, the demand of a point $p \in P^*_i$ will be its residual demand $d_p - c(*, \lps(p))$. Notice that the LP solution $\vec{x'}$ restricted to $\rs'$ is a feasible solution for both $\{ \pt^1_i\}_i$ and $\{ \pt^2_i\}_i$. We round the solution $\vec{x}'$ for these two instances of $\prc$ problem separately, and take their union.

\subsubsection{Heavy Instance and Hyper-4cuboid Covering Problem}

 We round the heavy instance $(\{\pt^1_i\}_i, \rs')$ by reducing it to a weighted geometric set cover problem called hyper-4cuboid covering problem (\hcp). 
 \hcp is a generalization of the R3U problem considered in \cite{BansalP10}.
 In \hcp, we are given a set  $U$ of points $p = (p_1, p_2, p_3, p_4)$ in 4-dimensional space. We are also given a set $M$ of geometric objects. Each object $v \in M$ is a set of $m$ disjoint 4-dimensional cuboids. Formally, $v = \{r_1, r_2, \ldots r_m\}$, where each $r_i$ is a 4-dimensional cuboid (4-cuboid) of the form $[2i, 2i+1] \times [0, x] \times [y_1, y_2] \times [0, z]$. We call $v$ a hyper-4cuboid. Each hyper-4cubiod has a cost $w(v)$. The goal is to find the minimum cost set $S \subseteq M$ such that $S$ covers all the points in $U$. We say that a hyper-4cuboid $v$ covers a point $p$ if $\exists r_i \in v, p \in r_i$.

Now we show that there is a reduction from the \prc problem on the heavy instance to the \hcp problem.  For a point $p \in \pt^1_i$ with coordinates $(x, y)$, we create a point $p' \in U$ with coordinates $(2i+1/2, x,y, \td_p)$. Note that the last coordinate of $p'$ is determined by $\td_p$, which  is the residual demand of point $p$. And the first coordinate of $p'$ is determined by the index of the set $\pt^1_i$. For every rectangle $r \in \rs'$, we create a hyper-4cuboid  $v_r \in M$, which contains exactly one 4cuboid for each point set $\pt^1_i$. For a given index $i \in [m]$ and a rectangle $r = [0, x] \times [y_1, y_2]$, there is a 4cuboid $r_i$ of the form $[2i, 2i+1] \times [0, x] \times [y_1, y_2] \times [0, \tc(i,r)]$. Note that the last coordinate is determined by the capacity of rectangle $r$ towards the point set $\pt^1_i$. The cost of hyper-4cuboid $v_r$ is same as the cost of rectangle $r$.

\begin{lemma}
	\label{lem:prc2hpc}
Suppose there is a feasible solution of cost $\alpha$ to the \prc problem on the heavy instance, where for each $i \in [m]$ and for each $p \in \pt^1_i$, the demand of the point $p$ is completely satisfied by the rectangles of class at least the class of $p$. Then there is a feasible solution of cost at most $\alpha$ for  the corresponding instance of \hcp. Similarly, a solution of cost $\alpha$ to the \hcp problem gives a solution of the same cost for the heavy instance of \prc.
\end{lemma}

\begin{proof}
Suppose $S \subset \rs'$ is a feasible solution of cost $\alpha$ satisfying the condition stated in the lemma. Then, we claim that the set of hyper-4 cuboids $S'$ corresponding to the rectangles $r \in S$ is a feasible solution to \hcp. Consider a point $p' \in U$  in the instance of \hcp produced by our reduction. Suppose $p'$ has coordinates  $(2i+1/2, x,y, \td_p)$. Then, there must be a point $p \in \pt^1_i$ with coordinates $(x,y)$ and the demand $\td_p$. Suppose $r \in S$ covers this point $p$, and has dimensions $r = [0, x] \times [y_1, y_2]$. Then, we claim  that $v_r \in S'$ covers the point $p'$. This is true, since $v_r$ contains a 4-cuboid with coordinates $h = [2i, 2i+1] \times [0, x] \times [y_1, y_2] \times [0, \tc(i,r)]$. It is easy to verify that $p' \in h $ as $\tc(i,r) > \td_p$, which follows from the condition of the lemma.

The opposite direction also follows from the one-to-one correspondence between the rectangles in $\rs$ and the hyper-4cuboids. Suppose $S'$ is a feasible solution to the instance of \hcp problem. Then, it is easy to verify that picking the rectangles $r \in \mathcal{R}$ corresponding to the hyber-4cuboids $v_r \in S$ defines a feasible a solution to $(\{\pt^1_i\}_i, \rs')$.  
\end{proof}

Now, observe crucially that there is a fractional solution of cost at most the cost of LP solution to the heavy instance satisfying the requirements of Lemma \ref{lem:prc2hpc}. This is true because the LP solution $\vec{x}'$ satisfies the inequality Eq.(\ref{eq:heavypoint}).  Therefore to apply Theorem \ref{thm:UC}, it remains to quantify the union complexity of hyper-4cuboids in our reduction.

\begin{lemma}
	\label{lem:4cubeUC}
	The union complexity of any $q$ hyper-4 cuboids in $M$ is at most $O(q m \log P)$, where $P = \frac{\max \{p_{ij}\}}{\min \{p_{ij}\}}$.
\end{lemma}

\begin{proof}
	Consider a set $S \subseteq M$ of $q$ hyper-4 cuboids. For $i = 1, 2, \ldots m$, define $\mathcal{L}_i$ as the set of all 4-cuboids at level $i$; Formally, $\mathcal{L}_i \subseteq \{\bigcup_{i, v \in S} r_i(v) \}$, such that every 4-cuboid $r \in \mathcal{L}_i$ has the first dimension equal to $[2i, 2i+1]$. In other words, $\mathcal{L}_i$ is the set of all 4-cuboids which have the same first dimension $[2i, 2i+1]$. Now, observe from our construction that for any pair $i,i'$, the objects in  $\mathcal{L}_i$ and  $\mathcal{L}_{i'}$ do not intersect. This is because 4-cuboids at different levels are separated by a distance of 1 in the first dimension.
	Therefore, the union complexity of the subset $S$ is equal to $m$ times the maximum of union complexity of 4-cuboids in $\{\mathcal{L}_i\}_i$. However, 4-cuboids in the sets ${\mathcal{L}_i}$ all have the same form; thus, it suffices to bound the union complexity for any ${\mathcal{L}_i}$. 
	
	Fix some $i$, and consider the 
	4-cuboids in the set ${\mathcal{L}_i}$. Notice carefully that all the 4-cuboids in ${\mathcal{L}_i}$ share the same first dimension $[2i, 2i+1]$. This implies that we can ignore the first dimension as it does not add any edges to the arrangement of objects in ${\mathcal{L}_i}$. Now consider the projection of 
	4-cuboids to the  remaining 3 dimensions; These are cuboids of the form $[0,x] \times [y_1, y_2] \times [0,z]$. For these type of cuboids, \cite{BansalP10} proves that the union complexity is at most $O(q\theta)$, where $\theta$ is the number of distinct values taken by $z$. In our reduction, the values of $z$ correspond to the capacities of rectangles. Recall that the capacities of rectangles in \prc are defined based on the processing lengths of operations.
	As we round down the capacities to the nearest power of 2, the number of distinct values taken by $z$ is at most $\log P$, where $P = \frac{\max \{p_{ij}\}}{\min \{p_{ij}\}}$.
	Putting everything together, we complete the proof.
\end{proof}

From Lemma \ref{lem:prc2hpc}, \ref{lem:4cubeUC} and Theorem \ref{thm:UC}, we get the following.

\begin{lemma}
	\label{lem:costHeavy}
	There is a solution of cost $O(\log (m \log P))$ times the cost of LP solution for the \hcp problem, and hence for the \prc on the heavy instance.
\end{lemma}

\subsection{Covering Light Points and Geometric Multi Cover by Hyper-cuboids}

Here, we design a $O(\log m)$ approximation algorithm to \prc problem restricted to the light instance: $(\{\pt^2_i \}_i, \rs')$. There is one main technical difference between our algorithm for the light case from the algorithm of Bansal and Pruhs \cite{BansalP10}. Bansal and Pruhs divide the light instance into $\log P$ different levels, where in each level they solve 2-dimensional geometric multi-cover problem, and show a $O(1)$ approximation to each level. This in turn leads to a constant factor approximation algorithm to the light instance. However, in our problem we cannot solve the instances separately. Therefore, using a simple trick, we map our problem into a  single instance of  3-dimensional problem, and show a $(\log m)$ approximation to it. This is precisely the reason we loose the factor $O(\log m)$ for the light case, as the union complexity of our objects becomes $m$ times the union complexity of the objects in \cite{BansalP10}.

Recall that for every point $p \in \pt^2_i$, the LP solution satisfies the inequality Eq.(\ref{eq:lightpoint}). Again, our idea is to reduce the instance to a geometric uncapacitated set multi-cover problem, and then appeal to Theorem \ref{thm:UC}. 
The geometric objects produced by our reduction are {\em sets of cuboids}. Hence we abbreviate this problem as \gmc.

In \gmc, we are given a set $U$ of points in  3-dimensional space, and a set $M$ of geometric objects. Each geometric object $v \in M$ is a collection of $m$ disjoint cuboids, which we call as hyper-cuboids. Each point $p \in U$ has a demand $e_p$, and each hyper-cuboid $v \in M$ has a cost $b_v$. The individual cuboids constituting a hyper-cuboid $v$ have the form $[x, x+1] \times [0, y] \times [z_1, z_2]$.  The objective is to find the minimum cost subset $S \subseteq M$ of hyper-cuboids such that for every point $p \in U$, there are at least $e_p$ number of hyper-cuboids in $S$ that contain the point $p$. 

\smallskip
Now we give the reduction $\mathcal{H}$ from an instance of \prc to an instance of \gmc. Fix some $i \in [m]$. Let $T > 0$ be some large constant, and recall $P = \frac{\max \{p_{ij}\}}{\min \{p_{ij}\}}$.
For every light point $p \in \pt^2_i$ with coordinates $(x_p, y_p)$, we create $\log P$ different points, each of them shifted in the second dimension by a distance of $T$.
The $\log P$ different points corresponding to a point $p \in \pt^2_i$ have coordinates  $(2i+1/2, kT + x_p, y_p)$, for $k = 1, 2,...\log P$. To complete the description of point set $U$, we need to define the demands of each point, which we will do after describing the cuboids in $M$.

For every rectangle $r \in \rs'$, with dimensions $[0, r_j] \times [t_1, t_2]$, we create one hyper-cuboid $v_r \in M$. The cost of hyper-cuboid $v_r$ is same the as the cost of rectangle $r$.
The hyper-cuboid $v_r$ contains $m$ different cuboids, one for each point set $\pt^2_1, \pt^2_2, \ldots \pt^2_m$. Fix an index $i \in [m]$. The cuboid corresponding to the set $\pt^2_i$ has dimensions $[2i, 2i+1] \times [0, kT + r_j] \times [t_1, t_2]$, where $k$ denotes the class of rectangle $r$ with respect to $\pt^2_i$.  Note that there is one-to-one correspondence between rectangles $r \in \rs'$ to hyper-cuboids $v_r \in M$. Let $\mathcal{H}^{-1}(v_r) = r$; that is, the rectangle $r \in \rs'$ that corresponds to the hyper-cuboid $v_r \in M$. We say that $v$ is picked to an extent of $x'_r$ in the LP solution $\vec{x'}$ to mean the extent to which the rectangle $r$ is picked in the LP solution. 
   
Fix a point $p \in U$. We say that $p$ is contained in the hyper-cuboid $v \in M$, if it is contained in any of the $m$ cuboids that constitute $v$.  Now we are ready to define the demand $e_p$ of a point $p$ as $ e_p = \left \lfloor \sum_{v: p \in v} x_{\mathcal{H}^{-1}(v)} \right \rfloor$.

To understand the above definition, let us consider a point $p \in U$ with coordinates $(2i+1/2, kT + x, y)$. The cuboids that can contain $p$ should have form $[2i, 2i+1] \times [0, kT + r_j] \times [t_1, t_2]$, where $kT + x < r_j$ and $ y \in [t_1, t_2]$. These cuboids correspond to the class $k$ rectangles in $\rs$. Therefore, demand of a point $p \in U$ with coordinates $(2i+1/2, kT + x, y)$ is exactly equal to the number of class $k$ rectangles that cover the point $p \in \pt^2_i$ in the LP solution.

This immediately tells us that the LP solution $\vec{x'}$ is a feasible fractional solution to the instance of \gmc problem produced by our reduction $\mathcal{H}$. Now, we show the opposite direction.

\begin{lemma}
\label{l:lighteasy}
If there is an integral solution to the instance of \gmc produced by our reduction $\mathcal{H}$ of cost at most $\alpha$ times the LP cost, then there is an integral solution of the same cost to the instance of $\prc$ on the light instance.
\end{lemma}

\begin{proof}
Suppose $S \subset M$ is a solution to \gmc problem. We show that there is a corresponding solution to \prc problem of exactly the same cost.
Our solution is simple. For every hyper-cuboid $v \in S$, we pick the corresponding rectangle $r$ in our solution $S'$ to $\prc$. From our reduction $\mathcal{H}$, the costs of hyper-cuboids is same as the costs of rectangles, which implies that solution costs of the two problems are the same. It remains to show that $S'$ is a feasible solution for the light instance.

Fix a point set $\pt_i$, and consider any point $p \in \pt_i$. Let $p = (x,y)$. In our reduction, there are $\log P$ different points $p'_1, p'_2,..., p'_{\log P}$, where the point $p'_k$ has coordinates $(2i+1/2, kT+x, y)$. The demand of $d_{p'_k}$ is exactly equal to the total number of class $k$ rectangles that cover the point $p$ in the LP solution. Recall that in our reduction $\mathcal{H}$, each rectangle $r \in \rs$ produces exactly one cubiod at the level $i$ (that is with first dimension $[2i, 2i+1]$) in the corresponding hyper-cuboid $v_r \in S'$. By abusing notation, let us call denote this cuboid by $h(r,i)$. 

Let class of point $p$ be $k^*$. The total capacity of rectangles that cover point $p \in \pt_i$ in our solutin $S'$ is
\begin{eqnarray}
\sum_{k \leq k^*} \sum_{r \in S': r \in \text{cls}(k)} 2^k \cdot \tc_{\min}(i)  
\end{eqnarray}

From the one-to-one correspondence between the rectangles and the cuboids, the term $\sum_{k \leq k^*} \sum_{r \in S': r \in \text{cls}(k)} 2^k \cdot \tc_{\min}(i) $ is exactly equal to the number of cuboids that cover the point $p'_k$, which is at least the demand $d_{p'_k}$. Now, in our reduction, the the demand of the point $p'_k$ (which has coordinates $(2i+1/2, kT+x, y)$) is exactly equal to the rounded down value of the total number of class $k$ rectangles that cover the point $p$ in the LP solution. Therefore, 
\begin{eqnarray*}
\sum_{k \leq k^*} \sum_{r \in S': r \in \text{cls}(k)} 2^k \cdot \tc_{\min}(i) \geq 
\sum_{k \leq k^*} 2^k \cdot \tc_{\min}(i)  \cdot \left \lfloor\sum_{r \in S': r \in \text{cls}(k)} x'_r \right \rfloor. 
\end{eqnarray*}

Simplifying the righthand side of the above equation, 
\begin{align*}
\sum_{k \leq k^*} 2^k \cdot \tc_{\min}(i)  \cdot \left \lfloor \sum_{r \in S': r \in \text{cls}(k)} x'_r \right \rfloor 
& \geq \sum_{k \leq k^*} 2^k \cdot \tc_{\min}(i) \cdot \left( \left(\sum_{r \in S': r \in \text{cls}(k)} x'_r \right) - 1\right) \\
&\geq \sum_{k \leq k^*} 2^k \cdot \tc_{\min}(i)  \cdot \left(\sum_{r \in S': r \in \text{cls}(k)} x'_r \right)  - 2^{k^*} \tc_{\min}(i) \\
& \geq 2d'(p) - d'(p) = d'(p), 
\end{align*}
where the last inequality follows from the Eq.(\ref{eq:lightpoint}), and the fact that $p$ is a light point. Therefore, $S'$ is a feasible solution to \gmc.
\end{proof}

Thus, to complete our algorithm for the light instance, it remains to design an $O(\log m)$ approximation algorithm to \gmc. For that we once again plan to rely on the Theorem \ref{thm:UC}, which requires to us to bound the union complexity of objects in $M$.

\begin{lemma}
\label{l:ucHyperCuboids}
The union complexity of any $q$ hyper-cuboids in $M$ is at most $O(mq)$.
\end{lemma}

\begin{proof}
Consider any $S \subseteq M$ of $q$ hyper-cuboids. From our definition, each hyper-cuboid $v$ is a set of $m$ cuboids. For $i =1, 2, ...,m$, define $\mathcal{L}_i$ as the set of all cuboids with first dimension $[2i, 2i+1]$. Clearly, for any two indices $i, i'$, the cuboids in $\mathcal{L}_i$ and $\mathcal{L}_{i'}$ do not intersect, as they are separated by a distance of 1 in the first dimension. Moreover, the cuboids in different $\mathcal{L}_{i}$ have same form except that they are shifted in the first dimension. Thus, the union complexity of $q$ hyper-cuboids in $S$ is at most $m$ times the union complexity of $q$ cuboids in $\mathcal{L}_{i}$, for any $i \in [m]$. 

Fix some $i$ and consider the set of $q$ cuboids in $\mathcal{L}_{i}$. Since all of the $q$ cuboids have exactly the same side in the first dimension, $[2i, 2i+1]$, the union complexity of $q$ cuboids in $\mathcal{L}_{i}$ is equal to the union complexity of the projection of the cuboids to the last two dimensions. This projection of the cubiods produces a set of axis parallel rectangles. Furthermore, these $q$ axis parallel rectangles partition into $\log P$ sets, $Z_1, Z_2, \ldots Z_{\log P}$, such that no two rectangles from different sets intersect. This is true as these rectangles have coordinates $[0,kT]$ in the second dimension. In each set $Z_k$, the rectangles are abutting the $Y$-axis (or the axis of second dimension). 
Bansal and Pruhs \cite{BansalP10} showed that the union of complexity such axis parallel rectangles abutting $Y$-axis is at most $O(|Z_k|)$. Therefore, the union complexity all the $q$ axis parallel rectangles is at most $\sum_{k}|Z_k| = O(q)$.   Thus we conclude that the union complexity $q$ hyper-cuboids is equal to $mq$, which completes the proof.
\end{proof}

Thus, from Theorem \ref{thm:UC}, we get a $O(\log m)$ approximation algorithm for the light instance of \prc.

\subsubsection{Proof of Theorem \ref{thm:OSUB}}
Our final solution for \prc problem is obtained by taking the union of all the rectangles picked in our solutions for the heavy and the light instance, and the set $\mathcal{S}$. Recall that in $\mathcal{S}$, we pick all the rectangles for which $x_r > 1/\beta$, for $\beta = 12$. The cost of our solution is at most $(O(\log (m \log P)) + O(\log m) + 1/\beta ))$ times the cost of  LP (\ref{lp:prc} - \ref{e:lpprcnonnegativity}), which implies an $O(\log (m \log P))$ approximation algorithm. This completes the proof. 


%% file: precedence.tex
\section{Precedence Constrained Scheduling}
\label{sec:mig}
In this problem, each job $j$ has a processing length $p_j > 0$, a release time $r_j > 0$. We have a set of $m$ identical machines; Each job
$j$ must be scheduled on exactly one of the $m$ machines. The important feature of the problem is that they are precedence constraints between jobs that capture the computational dependencies across jobs. The precedence constraints are  given by a partial order ``$\prec$", where a constraint
$j \prec j'$ requires that job $j'$ can only start after job $j$ is completed. Without loss of generality we assume that $r_{j'} > r_{j}$ whenever $j \prec j'$.
Our goal is to minimize the sum of arbitrary cost functions of completion times of jobs. Let $g_j(t)$ denote the cost of completing a job $j$ at time $t$.

Our algorithm for \pc consists of two main steps: In the first step, using a linear programming relaxation for the problem, we find a migratory schedule in which a job may be processed on multiple machines. In the second step, using some known results and some extra speed augmentation, we convert the migratory schedule into a non-migratory schedule in which every job is completely processed on a single machine. 

\subsection{Migratory Schedule}

First we write a time-indexed linear programming relaxation for the problem; using standard arguments, this can be converted into a polynomial sized LP by increasing the cost of solution by an $\epsilon$-factor for any $\epsilon > 0$. See \cite{Li17} for more details. Consider the LP (\ref{lp:Precobj} - \ref{e:PrecNN}) given below. The variables of our LP are $x_{jt}$, which indicate if a job $j$ is processed at time $t$. Note that we do not create $x_{ijt}$ variables indicating on which machine $j$ got scheduled at each time step. This is not required as machines are identical.
\begin{align}
	\textstyle  \text{Minimize} \quad \sum_{j} \sum_{t: t \geq r_j} x_{jt} \cdot g_j(t)/p_j +   \sum_{j} g_j(p_j) \label{lp:Precobj}  \vspace{-3mm} 
\end{align}
\vspace{-5mm}
\begin{align}
	\textstyle   \forall j&: &\qquad \sum_{t: t \geq r_j}  x_{jt}/p_j  &\geq 1 \label{e:PrecService} \\
	\textstyle   \forall t&: &\qquad \sum_{j: r_j \geq t} x_{jt} &\leq m    \label{eqn:PrecCapacity} \\
	\textstyle   \forall j&: &\qquad c_j &\geq \sum_{t : t \geq r_j} x_{jt} \cdot \left( \frac{t}{p_j} + \frac{1}{2} \right)   \label{e:PrecCost} \\
	\textstyle   \forall j, j', j \prec j'&: &\qquad   c_{j'} &\geq c_j + p_{j}    \label{eqn:Prec} \\
	\textstyle   \forall j,t&: &\qquad x_{jt} &\in [0,1] \label{e:PrecNN} 
\end{align}

The constraints (\ref{e:PrecService}) ensure that every job is completely processed. At any time $t$, at most $m$ units of jobs can be processed by $m$ machines; this is ensured by the constraints (\ref{eqn:PrecCapacity}). The variables $c_j$ are axillary, and indicate the completion time of job $j$. The constraints 
(\ref{eqn:Prec}) ensure that if $j \prec j'$, the completion time of $j'$ has to be at least the completion time of $j + p_{j'}$. Now consider the objective function. The first term lowerbounds the cost incurred by any job $j$ in a feasible schedule by its fractional cost. Suppose completion time of a job $j$ in an integral solution is $C_j$. Then, $\sum_{t: t \geq r_j} x_{jt} \cdot g_j(t)/p_j \leq \sum_{t: t \geq r_j} x_{jt} \cdot g_j(C_j)/p_j \leq g(C_j) \cdot \sum_{t: t \geq r_j} x_{jt}/p_j = g_j(C_j).$ The second term in the objective function is also clearly a lowerbound on the cost of any feasible schedule. Therefore, the cost of LP solution is at most twice the optimal solution.

\medskip

\noindent \textbf{Rounding.}
Let $\vec{x^*} : = \{ x^*_{ij} \}_{j,t}$ denote an optimal solution to LP (\ref{lp:Precobj} - \ref{e:PrecNN}). Define $C_j$  as the earliest time  in $\vec{x^*}$ when at least $p_j/2$ units of job $j$ is processed. We first produce a schedule where we only process $p_j/2$ units of each job $j$; then using speed augmentation we convert this into a valid schedule by processing twice the amount of each job that is scheduled at each time slot. Let $p'_j = p_j/2$.

\begin{claim}
\label{claim:truncate}
The cost incurred by a job $j$ in $\vec{x^*}$ is at least $g_j(C_j)/2$.
\end{claim}
\begin{proof}
Proof follows from a simple observation that $\sum_{t: t \geq C_j} x_{jt} \cdot g_j(t)/p_j \geq \sum^{C_j}_{t: t \geq r_j} x_{jt} \cdot g_j(C_j)/p_j = g_j(C_j)/2$.
\end{proof}

Next, we prove the following crucial property about the LP solution $\vec{x^*}$.
\begin{lemma}
\label{lem:property}
The vector of  $(C_j)_{j \in J}$ values satisfies the following property. For any time interval $[a, b]$,
$$
		\sum_{j: C_j \leq b,  r_j +  p'_j > a} \min\set{ p'_j,  r_j +  p'_j - a} \leq 2 \cdot (b-a).
$$
\end{lemma}

\begin{proof}
Consider LHS of the above equation. We classify the  jobs in the set $\{j: C_j \leq b,  r_j +  p'_j > a\}$ into two types: Let $J_1$ be the set of jobs such that $\forall j \in J_1, [r_j, C_j] \in [a,b]$. Let $J_2$ denote the remaining jobs: $J_2 = \{j: C_j \leq b,  r_j +  p'_j > a\} \setminus J_1$. Since $C_j \leq b$ for a job $ j \in J_1$, it implies that $p'_j$ units of the job $j$ is completely processed in the interval $[a,b]$.  Therefore,  from the capacity constraints (\ref{eqn:PrecCapacity}),  $\sum_{j \in J_1} p'_j \leq (b-a)$.  Now let us focus on the jobs in $J_2$. For every job $j \in J_2$, it is true that $r_j \leq a$. Further, since $r_j + p'_j > a$, it must be the case that at least $(r_j + p'_j - a)$ units of job $j$ is processed in the interval $[a, b]$. Again, from the capacity constraints (\ref{eqn:PrecCapacity}), total units of processing done on jobs in $J_2$ is at most $b-a$. Thus, we  conclude that $\sum_{j \in J_2} (r_j + p'_j - a) \leq (b-a)$.
Now, 
\begin{eqnarray*}
\sum_{j: C_j \leq b,  r_j +  p'_j > a} \min\set{ p'_j,  r_j +  p'_j - a} = 
\sum_{j \in J_1} p'_j \leq (b-a) + \sum_{j \in J_2} (r_j + p'_j - a)   \leq 2 (b-a),
\end{eqnarray*}
and the proof is complete.
\end{proof}

\begin{algorithm}
	\caption{\textsf{List-scheduling For Identical Machines}$\big(\alpha, (C_j)_{j \in J}\big)$}
	\begin{algorithmic}[1]
		\STATE let $c:[0, \infty) \to \set{0, 1, 2, \cdots, m}$ be identically 0
		\STATE for $j \in J$ in non-decreasing order of $C_j$ values, breaking ties so as to maintain precedence constraints
		\STATE \hspace{\algorithmicindent} let $\tilde S_j \gets \max\set{r_j, \max_{j' \prec j}\tilde C_{j'}}$
		\STATE \hspace{\algorithmicindent} let $\tilde C_j \geq \tilde S_j$ be the minimum $t'$ such that $\int_{t = \tilde S_j}^{t'} \mathbf{1}\set{c(t) < m} \sfd t =  p'_j/\alpha$
		\STATE \hspace{\algorithmicindent} schedule $j$ preemptively at time points $\set{t \in \left(\tilde S_j, \tilde C_j\right]: c(t)  < m}$
		\STATE \hspace{\algorithmicindent} let $c(t) \gets c(t) + 1$ for every $t \in \left(\tilde S_j, \tilde C_j\right]$ with $c(t) < m$
	\end{algorithmic}
\end{algorithm}

Our algorithm for the precedence constrained scheduling problem is simple: We schedule the jobs in the increasing order of $C_j$  values respecting the precedence and the capacity constraints. Below, we give a formal description the algorithm. Let the speed augmentation factor  be $\alpha \geq 1$.
In our algorithm, $\tilde S_j$ will be the time when $j$ is ready for scheduling: the earliest time when $j$ has arrived and all predecessors of $j$ have completed. $\tilde C_j$ is the completion time of $j$.

\begin{lemma} Suppose the vector $C$ satisfies condition in Lemma (\ref{lem:property}). Then the algorithm above  with $\alpha = 3$ constructs a schedule with $\tilde C_j \leq C_j$ for every $j \in J$. 
\end{lemma}


Till the end of this section, we fix an arbitrary $j^* \in J$. We shall show that $\tilde C_{j^*} \leq C_{j^*}$. Let $\mathsf{SOL}$ be the schedule constructed by the algorithm after the iteration $j^*$; let $J'$ be the set of jobs scheduled in $\mathsf{SOL}$; i.e, the set of jobs considered before $j^*$ (including $j^*$).   We say a time point $t$ is {\em busy} if all machines are processing jobs at time $t$ in the schedule $\mathsf{SOL}$; otherwise $t$ is {\em idle}.

For every $j \in J'$, let $ p_j(t)$ be the amount of job $j$ that has been processed by time $t$ in $\mathsf{SOL}$. 
Let $R_t = \set{j \in J': t \in (\tilde S_j, \tilde C_j]}$ be the set of unfinished jobs at time $t$ that are ready for scheduling, i.e, the set of arrived and incomplete jobs whose predecessors are completed.

\medskip

For every $t\in [0, \infty)$, let $f(t) = \min\set{t, \min_{j \in R_t}( r_j +  p_j(t))}$.  
\begin{claim}
	$f$ is a non-decreasing function.
\end{claim}
\begin{proof}
	We consider the change of $\min\set{t, \min_{j \in R_t}( r_j +  p_j(t))}$ as $t$ increases.
	The term $ r_j +  p_j(t)$ can only increase as $t$ increases. Jobs maybe added to and removed from $R_t$. Removing jobs from $R_t$ can only increase the minimum. So, it suffices to consider adding jobs to $R_t$. If some job $j'$ has $\tilde S_{j'} = t$, then $j' \notin A_t$ but $j' \in A_{t'}$ for $t' \in (t, t+1]$. We have either $t =  r_{j'}$ or $t = \tilde C_{j''}$ for some $j'' \prec j'$; notice that $j'' \in R_t$ in the latter case.  In the former case, $ r_{j'} +  p_{j'} (t) = t$; in the latter case, $ r_{j'} +  p_{j'}(t) =  r_{j'} \geq r_{j''} +  p_{j''} = r_{j''} +  p_{j''}(t)$, where the inequality is guaranteed by the preprocessing of the instance.  So, we always have $\min\set{t, \min_{j \in R_t \cup \{j'\}}( r_j +  p_j(t))} = \min\set{t, \min_{j \in R_t}( r_j +  p_j(t))}$, i.e, adding $j'$ to $R_t$ does not change the quantity.
\end{proof}

\begin{definition}
	For any $t \geq 0$, we say $t$ is a fresh point if $f(t) = t$.
\end{definition}

\begin{lemma}
	\label{lem:volume}
	Let $a$ be a fresh point. Let $b > a$ be an integer such that there are no fresh points in $(a, b)$.  The total length of idle slots in $(a, b)$ is at most $(b-a)/\alpha$. 
\end{lemma}

\begin{proof}
	Since there are no fresh points in $(a, b)$, we have that $f(t) = \min_{j \in R_t}\{r_j+ p_j(t)\}$ for every $t \in (a, b)$. If $t \in (a, b)$ is idle, then all the jobs in $R_t$ must be processed at time $t$: a job $j \in R_t$ is not being processed at time $t$ will contradict the fact that $t$ is an idle slot.  Let $(a', b') \subseteq (a, b)$ be a subset of idle slots such that $R_t$ is the same for all $t \in (a', b')$. Then $f(t)$ will increase at a rate of $\alpha$ within the interval $(a', b')$. Since $R_t$ will be changed only finite number of times as $t$ goes from $a$ to $b$, we have that $f(b) - f(a)$ is at least $\alpha$ times the total length of idle slots in $(a, b]$. Since $f(b) - f(a) = f(b) - a \leq b - a$, we have that the total length of idle slots in $(a, b]$ is at most $(b-a)/\alpha$.
\end{proof}

\begin{lemma}
    \label{lem:completion}
	We have $\tilde C_{j^*} \leq C_{j^*}$.
\end{lemma}

\begin{proof}
	Let $t$ be the last fresh point before $\tilde C_{j^*}$.  By Lemma (\ref{lem:volume}), the total length of idle slots in $(t, \tilde C_{j^*}]$ is at most $(\tilde C_{j^*} - t)/\alpha$.   The total length of busy slots is at most 
	\begin{align*}
		\frac{1}{m} \cdot \left (\sum_{j: j_k = j, C_j \leq C_{j^*},  r_j +  p_j > t} \min\set{ p_j,  r_j +  p_j - t} \right) \leq  2  \cdot(C_{j^*} - t)/\alpha.
	\end{align*}
	
	 When $\alpha = 3$ we get $(\tilde C_{j^*} - t) \leq (\tilde C_{j^*} - t)/3 + 2 (C_{j^*} - t)/3$.  This implies $\tilde C_{j^*} \leq C_{j^*}$, which completes the proof.
\end{proof}

The above lemma and Claim (\ref{claim:truncate}) imply the following.

\begin{lemma}
\label{lem:migratoryschedule}
There is a migratory 6-speed $O(1)$-approximation algorithm for the \pc problem.
\end{lemma}

\begin{proof}
Note that our algorithm when given a speed augmentation of 3 produces a migratory schedule where the completion time of every job is at most $C_j$. Therefore, from Claim (\ref{claim:truncate}), the cost incurred by our algorithm is at most twice the cost of LP solution. However, in our schedule we only process $p'_j$ units of a job $j$. To convert this into a valid schedule, we need another round of speed augmentation of factor 2. Therefore, the total speed augmentation required is 6. This completes the proof.
\end{proof}

\subsection{Non-Migratory Schedule}
The schedule produced by our algorithm processes each job in the interval $[S_j, \tilde C_j]$, but a single job may be processed on multiple machines. 
To convert this migratory schedule into non-migratory schedule, we use the algorithm either in \cite{kalyanasundaram2001eliminating, ImM16} or in  \cite{Kulkarni17}. The result in  \cite{Kulkarni17} is more suitable for our purposes.
 
\begin{theorem}[\cite{Kulkarni17,ImM16}]
\label{thm:non-mig}
Given a set of jobs where each job $j$ has an associated interval $[s_j, d_j]$, if there is a migratory schedule on $m$ machines that completely processes the jobs in their intervals, then there is a non-migratory schedule that schedules all the jobs within their intervals  that processes each job exactly on one machine if each machine is  given an $O(1)$ speed augmentation.
\end{theorem} 

Using the above theorem, we complete the proof of Theorem (\ref{thm:PCUB}).

\begin{proof}[Proof of Theorem (\ref{thm:PCUB})]
Let $S$ be the migratory scheduled produced by Lemma (\ref{lem:migratoryschedule}). We invoke the Theorem (\ref{thm:non-mig}) on $S$, where we set the interval for each job $j$ to be $[S_j, \tilde C_j]$.  Therefore, with an overall speed augmentation of $O(1)$, we get a schedule that completely processes each job in the interval $[S_j, \tilde C_j]$ on a single machine.  Further, the cost of the schedule  is at most twice the cost of the LP solution. Combining all the pieces, we complete the proof of Theorem (\ref{thm:PCUB}) regarding the upper bound. We show the lowerbound on the speed augmentation required in the next section.
\end{proof}

\subsection{Necessity of Speed-Augmentation and Proof of Theorem (\ref{thm:PCLB})}
Here we show that our analysis of \pc is essentially tight, and in particular $O(1)$ speed augmentation is necessary.  First we show the necessity of speed augmentation to the case of single machine. Our reduction is from the Densest-k-Subgraph (\dks) problem. In the \dks problem, we are given a graph $G$ with $n$ vertices and $m$ edges, and integers $k > 0, L >0$. The goal is to find a subgraph of $G$ on $k$ vertices that contains at least $L$ edges.
Recently, Manurangsi \cite{Manurangsi17} showed the following result.

\begin{theorem}
\label{thm:manu}
There exists a constant $c > 0$ such that, assuming exponential time hypothesis, no polynomial time algorithm can distinguish between the following two cases:
\begin{itemize}
\item There exists a set of $k$ vertices of $G$ that contains an induced subgraph with $L$ edges.
\item Every subset of $k$ vertices of $G$ contains at most $L/n^{1/(\log \log n)^c}$ induced edges.
\end{itemize}
\end{theorem}

For our reduction, it will be more convenient to consider the following  {\em vertex version} of the \dks problem. In this problem, given $G, L$, the goal is to approximate the number of  vertices $S$ of $G$ that contains at least $L$ edges.  An $\alpha$-approximation to the vertex version of \dks problem gives an $1/2\alpha^2$-approximation to the \dks problem. To see this, assume that there is an $\alpha$ approximation to vertex version of \dks. Given a \dks instance, we assume we know the number $L$ of edges in the densest $k$-subgraph of $G$ (this can be achieved by guessing). Thus we can find a set of vertices of size $\alpha \cdot k$ that contains at least $L$ edges.  We then randomly choose $k$ out of the $\alpha k$ vertices. It is easy to see that in expectation, at least $Lk(k-1)/(\alpha k)(\alpha k-1)$ edges will be contained in the sub-graph induced by the $k$ chosen vertices. This is at least $L/2\alpha^2$ when $k \geq 2$. This process can be easily derandomized.

\begin{corollary}
\label{thm:manu2}
There exists a constant $c > 0$ such that, assuming exponential time hypothesis, no polynomial time algorithm can distinguish between the following two cases:
\begin{itemize}
\item There exists a set of $k$ vertices of $G$ that contains an induced subgraph with $L$ edges.
\item Every subset of size $k \cdot n^{1/(\log \log n)^c}/2$ vertices of $G$ contains at most $L$ induced edges.
\end{itemize}
\end{corollary}

Now we give a reduction from the vertex version of \dks to the \pc problem on a single machine.

\begin{proof}[Proof of Theorem \ref{thm:PCLB}]

To keep the notation simple, we allow jobs to have weights in our proof. However, it easy to see that in the precedence constrained scheduling model, weighted and unweighted cases are equivalent. To see this, suppose there is a job $j$ of weight $w_j$. Then, one can create $w_j$ dummy jobs of zero processing length that all depend on job $j$. This ensures that until $j$ finishes, it pays $w_j$ cost towards the flow-time, which is same as job $j$ having a weight of $w_j$.

Consider an instance $I$ of the vertex version of \dks problem. Let $\mathcal{H}(I)$ denote the instance of \pc produced by our reduction. For every vertex $v \in G$, we create a job $j_v$ in $\mathcal{H}(I)$. The processing length of job $j_v$ is zero, and weight is 1. For every edge $e \in G$, we have a job $j_e \in \mathcal{H}(I)$. The processing length of job $j_e$ is 1 and its weight is zero. All these jobs have release time of 0.  If $v \in e$ in the graph, we have $j_e \prec j_v$.
Let $\delta = 1/n^2$, and $T > 0$ be an arbitrarily large number.  In the interval $(m-L, T]$, at every time step $m-L + k\delta$, for $k = 1, 2,...$, we release a job of size $\delta$ and weight $1$. Now consider the following two cases.

\medskip

\textbf{Case 1:} Suppose there exists a set of $k$ vertices $S$ of $G$ such that $G[S]$ contains at least $L$ edges. Then, we claim that there is a solution of cost at most $(n-k)(m-L) + k(T+L) + T - (m-L)$.  We can schedule the edge jobs correspondent to the edges not in $G[S]$ in interval $[0, m-L]$.  Then the $n-k$ vertex jobs $\set{j_v: v \notin S}$ can be scheduled at time $m-L$. All the edge jobs can be completed by time $T + L$; thus, the $k$  vertex jobs $\set{j_v: v \in S}$ can be scheduled at time $T + L$. The total weighted completion time for $\delta$-sized jobs is $T - (m - L)$.


\medskip 
\textbf{Case 2:} Every subset of size $k \cdot n^{1/2(\log \log n)^c}/2$ vertices of $G$ contains at most $L$ induced edges. This implies that in any algorithm at least $k \cdot n^{1/2(\log \log n)^c}/2$ vertex jobs are alive at time $t = m-L$. This further implies that at all times in the interval $(m-L, T]$ at least $k \cdot n^{1/2(\log \log n)^c}/2$ jobs of weight 1 are alive in any algorithm's schedule. This is because, if any job corresponding to an edge is scheduled, it will take at least one time  unit to complete that job. During that period at least $n^2$ jobs arrive which can not be scheduled. Note that $k \cdot n^{1/2(\log \log n)^c}/2 < n^2$. Therefore, cost incurred by any algorithm in this case is at least $(k \cdot n^{1/2(\log \log n)^c}/2) \cdot (T-(m-L))$.

\medskip
Suppose $T$ is some large polynomial in $n$. For our purposes $T = n^4$ suffices. In this case, the cost of schedule in case 1 is dominated by term $Tk$ and the cost the schedule in case 2 is dominated by $(k \cdot n^{1/2(\log \log n)^c}/2 \cdot) \cdot T$. From Corollary (\ref{thm:manu2}), no polynomial time algorithm can distinguish between these two cases. Clearly, our reduction $\mathcal{H}$ is of size polynomial in $n$ and takes polynomial time. This completes the proof.
\end{proof}

We now turn our attention to the multiple machines setting.  Here we give a simple reduction that shows  to obtain  a non-trivial approximation factor speed augmentation has to be at least the approximation factor of the underlying makespan minimization problem.  Our result holds not just for the identical machines case but also for any machine environment such as related machines setting, unrelated machines setting etc.

\begin{theorem}
\label{cl:speedup}
Suppose the precedence constrained scheduling to minimize makespan problem cannot be approximated better than $\gamma > 1$ for a machine environment $\mathcal{M}$. Then, unless given a speed augmentation of at least $\gamma$, the problem of minimizing the total unweighted flow-time cannot be approximated better than $\Omega(n^{1-c})$ for any constant $ c \in [0,1)$.
\end{theorem}

\begin{proof}
Consider an instance $\mathcal{I}$ of precedence constrained scheduling to minimize makespan problem. Let $m$ be the number of machines and $n'$ be the number of jobs in $\mathcal{I}$ . Let $J_1$ denote the set of jobs in $\mathcal{I}$.
Without loss of generality, assume that the optimal makespan is 1; this can be achieved by scaling the processing lengths of jobs by the optimal makespan value. With this assumption, we know that it is NP-hard to distinguish if optimal makespan is 1 or $\gamma$. We create an instance $\mathcal{I}'$ for the flow-time problem as follows: In  $\mathcal{I}'$, the machines are same as that in $\mathcal{I}$. Now we describe the job set: At time $t =0$, we release the jobs in the makespan instance $J_1$ with same precedence constraints. 

Let $T = \gamma^{\epsilon}$ and let $\delta > 0$ be a very tiny constant. For each time instant $t  =  (1 + k\delta) \in [1, T)$  a single job $j$ arrives.  The processing length of $j$ is such that it takes exactly of $\delta$ time units to complete the job on any machine. Further,  $j$ can be processed only if all the jobs preceding it are completed; i.e., $j' \prec j$ for all $j'$ with $r_{j'} < r_j$.  This completes the description of our instance $\mathcal{I}'$. Let $J_2$ denote the set of jobs that arrive in the interval $(1, T]$.

Now, let us calculate the total flow-time of jobs in the optimal solution.  The total flow-time of jobs in $J_1$ is at most $n'$ since we assumed that optimal makespan is 1. The total flow-time of jobs that are released in the interval $[1, T)$ is exactly equal to $T$ as there is exactly one job alive at any time $t \in [1, T)$. Thus, the total flow-time of jobs in the optimal solution is at most $n' + T$.

Let $\mathcal{A}$ be any polynomial time algorithm that is given a speed of $\gamma^{1-\epsilon}$ for any $\epsilon > 0$. Let us calculate the total flow-time of jobs in the schedule produced by $\mathcal{A}$. Towards that let us focus on the number of jobs that are alive at time $t = (\gamma^\epsilon - 1)$. We claim that exactly $(\gamma^\epsilon - 1)/\delta$ jobs from the set $J_2$ are alive at that time. This is true since  $\mathcal{A}$ cannot complete the jobs in the makepsan instance $J_1$ earlier than $\gamma^\epsilon$, and the jobs that are released in the interval $[1, \gamma^\epsilon]$ cannot begin their execution until all the jobs in $J_1$ are completed. The total number of jobs that arrive in $[1, \gamma^\epsilon]$ is equal to $(\gamma^\epsilon - 1)/\delta$.  Due to precedence constraints, $\mathcal{A}$ has to process jobs in $J_2$ only one at time.  This implies that at all times $t \in (T-1/2, T]$ at least $(\gamma^\epsilon - 1)/2\delta$ jobs are alive. Thus, the total flow-time of $\mathcal{A}$ is at least $(\gamma^\epsilon - 1)^2/2\delta$.

Suppose $1/\delta = n'^{1/c}$, where $n'$ is the number of jobs in the makespan instance. Since $c$ is a constant, this ensures the our reduction is of size  $O(n'^{1/c})$.  The ratio of the cost of $\mathcal{A}$ to that of the optimal solution is at least $((\gamma^\epsilon - 1)^2/2\delta + n')/ (n' +  \gamma^{\epsilon}) = \Omega(n'^{(1/c - 1}))$. Now, let us express this in terms of the number of jobs in the flow-time instance. The total number of jobs in the flow-time instance $n$ is equal to $n' + T/\delta \approx T/\delta = T \cdot n'^{(1/c)}$. Thus, writing our lowerbound in terms $n$ we get a  gap of $n^{1-c}$ on the approximation factor, which completes the proof.
\end{proof}

\begin{proof}[Proof of Theorem (\ref{thm:PCUB})]
Above theorem immediately gives the lowerbound on the speed augmentation factor claimed in the Theorem (\ref{thm:PCUB}).
\end{proof}

%% file: paper.bbl
\begin{thebibliography}{53}
\providecommand{\natexlab}[1]{#1}
\providecommand{\url}[1]{\texttt{#1}}
\expandafter\ifx\csname urlstyle\endcsname\relax
  \providecommand{\doi}[1]{doi: #1}\else
  \providecommand{\doi}{doi: \begingroup \urlstyle{rm}\Url}\fi

\bibitem[had()]{hadoop}
Apache hadoop.
\newblock In \emph{http://hadoop.apache.org/}.

\bibitem[Agrawal et~al.(2016)Agrawal, Li, Lu, and Moseley]{AgrawalLLM16}
Kunal Agrawal, Jing Li, Kefu Lu, and Benjamin Moseley.
\newblock Scheduling parallel {DAG} jobs online to minimize average flow time.
\newblock In \emph{Proceedings of the Twenty-Seventh Annual {ACM-SIAM}
  Symposium on Discrete Algorithms, {SODA} 2016, Arlington, VA, USA, January
  10-12, 2016}, pages 176--189, 2016.
\newblock \doi{10.1137/1.9781611974331.ch14}.
\newblock URL \url{https://doi.org/10.1137/1.9781611974331.ch14}.

\bibitem[Ahmadi et~al.(2005)Ahmadi, Bagchi, and Roemer]{ahmadi2005coordinated}
Reza Ahmadi, Uttarayan Bagchi, and Thomas~A Roemer.
\newblock Coordinated scheduling of customer orders for quick response.
\newblock \emph{Naval Research Logistics (NRL)}, 52\penalty0 (6):\penalty0
  493--512, 2005.

\bibitem[Ahmadi et~al.(2017)Ahmadi, Khuller, Purohit, and
  Yang]{ahmadi2017scheduling}
Saba Ahmadi, Samir Khuller, Manish Purohit, and Sheng Yang.
\newblock On scheduling coflows.
\newblock In \emph{International Conference on Integer Programming and
  Combinatorial Optimization}, pages 13--24. Springer, 2017.

\bibitem[Anand et~al.(2012)Anand, Garg, and Kumar]{AnandGK12}
S.~Anand, Naveen Garg, and Amit Kumar.
\newblock Resource augmentation for weighted flow-time explained by dual
  fitting.
\newblock In \emph{SODA}, pages 1228--1241, 2012.

\bibitem[Bansal and Khot(2010)]{bansal2010inapproximability}
Nikhil Bansal and Subhash Khot.
\newblock Inapproximability of hypergraph vertex cover and applications to
  scheduling problems.
\newblock \emph{Automata, Languages and Programming}, pages 250--261, 2010.

\bibitem[Bansal and Kulkarni(2015)]{BansalK15}
Nikhil Bansal and Janardhan Kulkarni.
\newblock Minimizing flow-time on unrelated machines.
\newblock In \emph{Proceedings of the Forty-Seventh Annual {ACM} on Symposium
  on Theory of Computing, {STOC} 2015, Portland, OR, USA, June 14-17, 2015},
  pages 851--860, 2015.
\newblock \doi{10.1145/2746539.2746601}.
\newblock URL \url{http://doi.acm.org/10.1145/2746539.2746601}.

\bibitem[Bansal and Pruhs(2010)]{BansalP10}
Nikhil Bansal and Kirk Pruhs.
\newblock The geometry of scheduling.
\newblock In \emph{IEEE Symposium on the Foundations of Computer Science},
  pages 407--414, 2010.

\bibitem[Bansal and Pruhs(2012)]{bansal2012weighted}
Nikhil Bansal and Kirk Pruhs.
\newblock Weighted geometric set multi-cover via quasi-uniform sampling.
\newblock In \emph{European Symposium on Algorithms}, pages 145--156. Springer,
  2012.

\bibitem[Bansal et~al.(2007)Bansal, Chan, Khandekar, Pruhs, Stein, and
  Schieber]{bansal2007non}
Nikhil Bansal, Ho-Leung Chan, Rohit Khandekar, Kirk Pruhs, Cliff Stein, and
  Baruch Schieber.
\newblock Non-preemptive min-sum scheduling with resource augmentation.
\newblock In \emph{Foundations of Computer Science, 2007. FOCS'07. 48th Annual
  IEEE Symposium on}, pages 614--624. IEEE, 2007.

\bibitem[Bansal et~al.(2010)Bansal, Krishnaswamy, and Nagarajan]{BansalKN10}
Nikhil Bansal, Ravishankar Krishnaswamy, and Viswanath Nagarajan.
\newblock Better scalable algorithms for broadcast scheduling.
\newblock In \emph{ICALP (1)}, pages 324--335, 2010.

\bibitem[Bansal et~al.(2014)Bansal, Charikar, Krishnaswamy, and
  Li]{bansal2014better}
Nikhil Bansal, Moses Charikar, Ravishankar Krishnaswamy, and Shi Li.
\newblock Better algorithms and hardness for broadcast scheduling via a
  discrepancy approach.
\newblock In \emph{Proceedings of the twenty-fifth annual ACM-SIAM symposium on
  Discrete algorithms}, pages 55--71. Society for Industrial and Applied
  Mathematics, 2014.

\bibitem[Batra et~al.(2018)Batra, Garg, and Kumar]{batra}
Jatin Batra, Naveen Garg, and Amit Kumar.
\newblock Constant factor approximation algorithm for weighted flow time on a
  single machine in pseudo-polynomial time.
\newblock \emph{CoRR}, abs/1802.07439, 2018.

\bibitem[Carr et~al.(2000)Carr, Fleischer, Leung, and
  Phillips]{carr2000strengthening}
Robert~D Carr, Lisa Fleischer, Vitus~J Leung, and Cynthia~A Phillips.
\newblock Strengthening integrality gaps for capacitated network design and
  covering problems.
\newblock In \emph{SODA}, pages 106--115, 2000.

\bibitem[Chan et~al.(2012)Chan, Grant, K{\"o}nemann, and
  Sharpe]{chan2012weighted}
Timothy~M Chan, Elyot Grant, Jochen K{\"o}nemann, and Malcolm Sharpe.
\newblock Weighted capacitated, priority, and geometric set cover via improved
  quasi-uniform sampling.
\newblock In \emph{Proceedings of the twenty-third annual ACM-SIAM symposium on
  Discrete Algorithms}, pages 1576--1585. Society for Industrial and Applied
  Mathematics, 2012.

\bibitem[Chen and Hall(2007)]{chen2007supply}
Zhi-Long Chen and Nicholas~G Hall.
\newblock Supply chain scheduling: Conflict and cooperation in assembly
  systems.
\newblock \emph{Operations Research}, 55\penalty0 (6):\penalty0 1072--1089,
  2007.

\bibitem[Chowdhury and Stoica(2012)]{chowdhury2012coflow}
Mosharaf Chowdhury and Ion Stoica.
\newblock Coflow: A networking abstraction for cluster applications.
\newblock In \emph{Proceedings of the 11th ACM Workshop on Hot Topics in
  Networks}, pages 31--36. ACM, 2012.

\bibitem[Chowdhury and Stoica(2015)]{chowdhury2015efficient}
Mosharaf Chowdhury and Ion Stoica.
\newblock Efficient coflow scheduling without prior knowledge.
\newblock In \emph{ACM SIGCOMM Computer Communication Review}, volume~45, pages
  393--406. ACM, 2015.

\bibitem[Chowdhury et~al.(2014)Chowdhury, Zhong, and
  Stoica]{chowdhury2014efficient}
Mosharaf Chowdhury, Yuan Zhong, and Ion Stoica.
\newblock Efficient coflow scheduling with varys.
\newblock In \emph{ACM SIGCOMM Computer Communication Review}, volume~44, pages
  443--454. ACM, 2014.

\bibitem[Devanur and Kulkarni(2017)]{Kulkarni17}
Nikhil Devanur and Janardhan Kulkarni.
\newblock A unified rounding algorithm for unrelated machines scheduling.
\newblock In \emph{Preprint,
  https://users.cs.duke.edu/~kulkarni/papers/tardines.pdf}, 2017.

\bibitem[Gangal and Ranade(2008)]{gangal2008precedence}
Devdatta Gangal and Abhiram Ranade.
\newblock Precedence constrained scheduling in (2- 73p+ 1)⋅ optimal.
\newblock \emph{Journal of Computer and System Sciences}, 74\penalty0
  (7):\penalty0 1139--1146, 2008.

\bibitem[Garg and Kumar(2007)]{GargK07}
Naveen Garg and Amit Kumar.
\newblock Minimizing average flow-time : Upper and lower bounds.
\newblock In \emph{FOCS}, pages 603--613, 2007.

\bibitem[Garg et~al.(2007)Garg, Kumar, and Pandit]{garg2007order}
Naveen Garg, Amit Kumar, and Vinayaka Pandit.
\newblock Order scheduling models: hardness and algorithms.
\newblock \emph{FSTTCS 2007: Foundations of Software Technology and Theoretical
  Computer Science}, pages 96--107, 2007.

\bibitem[Graham(1966)]{graham1966bounds}
Ronald~L Graham.
\newblock Bounds for certain multiprocessing anomalies.
\newblock \emph{Bell Labs Technical Journal}, 45\penalty0 (9):\penalty0
  1563--1581, 1966.

\bibitem[Grandl et~al.(2016{\natexlab{a}})Grandl, Chowdhury, Akella, and
  Ananthanarayanan]{grandl2016altruistic}
Robert Grandl, Mosharaf Chowdhury, Aditya Akella, and Ganesh Ananthanarayanan.
\newblock Altruistic scheduling in multi-resource clusters.
\newblock In \emph{OSDI}, pages 65--80, 2016{\natexlab{a}}.

\bibitem[Grandl et~al.(2016{\natexlab{b}})Grandl, Kandula, Rao, Akella, and
  Kulkarni]{GrandlKRAK16}
Robert Grandl, Srikanth Kandula, Sriram Rao, Aditya Akella, and Janardhan
  Kulkarni.
\newblock {GRAPHENE:} packing and dependency-aware scheduling for data-parallel
  clusters.
\newblock In \emph{12th {USENIX} Symposium on Operating Systems Design and
  Implementation, {OSDI} 2016, Savannah, GA, USA, November 2-4, 2016.}, pages
  81--97, 2016{\natexlab{b}}.
\newblock URL
  \url{https://www.usenix.org/conference/osdi16/technical-sessions/presentation/grandl_graphene}.

\bibitem[Im and Moseley(2011)]{ImM11}
Sungjin Im and Benjamin Moseley.
\newblock Online scalable algorithm for minimizing ;k-norms of weighted flow
  time on unrelated machines.
\newblock In \emph{SODA}, pages 95--108, 2011.

\bibitem[Im and Moseley(2012)]{ImM12}
Sungjin Im and Benjamin Moseley.
\newblock An online scalable algorithm for average flow time in broadcast
  scheduling.
\newblock \emph{ACM Transactions on Algorithms}, 8\penalty0 (4):\penalty0 39,
  2012.

\bibitem[Im and Moseley(2016)]{ImM16}
Sungjin Im and Benjamin Moseley.
\newblock General profit scheduling and the power of migration on heterogeneous
  machines.
\newblock In \emph{Proceedings of the 28th {ACM} Symposium on Parallelism in
  Algorithms and Architectures, {SPAA} 2016, Asilomar State Beach/Pacific
  Grove, CA, USA, July 11-13, 2016}, pages 165--173, 2016.

\bibitem[Im et~al.(2014{\natexlab{a}})Im, Kulkarni, and Munagala]{ImKM14}
Sungjin Im, Janardhan Kulkarni, and Kamesh Munagala.
\newblock Competitive algorithms from competitive equilibria: non-clairvoyant
  scheduling under polyhedral constraints.
\newblock In \emph{Symposium on Theory of Computing, {STOC} 2014, New York, NY,
  USA, May 31 - June 03, 2014}, pages 313--322, 2014{\natexlab{a}}.
\newblock \doi{10.1145/2591796.2591814}.
\newblock URL \url{http://doi.acm.org/10.1145/2591796.2591814}.

\bibitem[Im et~al.(2014{\natexlab{b}})Im, Kulkarni, Munagala, and
  Pruhs]{ImKMP14}
Sungjin Im, Janardhan Kulkarni, Kamesh Munagala, and Kirk Pruhs.
\newblock Selfishmigrate: {A} scalable algorithm for non-clairvoyantly
  scheduling heterogeneous processors.
\newblock In \emph{55th {IEEE} Annual Symposium on Foundations of Computer
  Science, {FOCS} 2014, Philadelphia, PA, USA, October 18-21, 2014}, pages
  531--540, 2014{\natexlab{b}}.
\newblock \doi{10.1109/FOCS.2014.63}.
\newblock URL \url{https://doi.org/10.1109/FOCS.2014.63}.

\bibitem[Im et~al.(2015{\natexlab{a}})Im, Kulkarni, and Munagala]{ImKM15}
Sungjin Im, Janardhan Kulkarni, and Kamesh Munagala.
\newblock Competitive flow time algorithms for polyhedral scheduling.
\newblock In \emph{{IEEE} 56th Annual Symposium on Foundations of Computer
  Science, {FOCS} 2015, Berkeley, CA, USA, 17-20 October, 2015}, pages
  506--524, 2015{\natexlab{a}}.
\newblock \doi{10.1109/FOCS.2015.38}.
\newblock URL \url{https://doi.org/10.1109/FOCS.2015.38}.

\bibitem[Im et~al.(2015{\natexlab{b}})Im, Li, Moseley, and
  Torng]{im2015dynamic}
Sungjin Im, Shi Li, Benjamin Moseley, and Eric Torng.
\newblock A dynamic programming framework for non-preemptive scheduling
  problems on multiple machines.
\newblock In \emph{Proceedings of the Twenty-Sixth Annual ACM-SIAM Symposium on
  Discrete Algorithms}, pages 1070--1086. Society for Industrial and Applied
  Mathematics, 2015{\natexlab{b}}.

\bibitem[Kalyanasundaram and Pruhs(2000)]{kirk}
Bala Kalyanasundaram and Kirk Pruhs.
\newblock Speed is as powerful as clairvoyance.
\newblock \emph{Journal of the ACM}, 47\penalty0 (4):\penalty0 617--643, 2000.

\bibitem[Kalyanasundaram and Pruhs(2001)]{kalyanasundaram2001eliminating}
Bala Kalyanasundaram and Kirk~R Pruhs.
\newblock Eliminating migration in multi-processor scheduling.
\newblock \emph{Journal of Algorithms}, 38\penalty0 (1):\penalty0 2--24, 2001.

\bibitem[Khuller and Purohit(2016)]{khuller2016brief}
Samir Khuller and Manish Purohit.
\newblock Brief announcement: Improved approximation algorithms for scheduling
  co-flows.
\newblock In \emph{Proceedings of the 28th ACM Symposium on Parallelism in
  Algorithms and Architectures}, pages 239--240. ACM, 2016.

\bibitem[Lam and Sethi(1977)]{lam1977worst}
Shui Lam and Ravi Sethi.
\newblock Worst case analysis of two scheduling algorithms.
\newblock \emph{SIAM Journal on Computing}, 6\penalty0 (3):\penalty0 518--536,
  1977.

\bibitem[Lenstra and Rinnooy~Kan(1978)]{LR78}
J.~K. Lenstra and A.~H.~G. Rinnooy~Kan.
\newblock Complexity of scheduling under precedence constraints.
\newblock \emph{Oper. Res.}, 26\penalty0 (1):\penalty0 22--35, February 1978.
\newblock ISSN 0030-364X.
\newblock \doi{10.1287/opre.26.1.22}.
\newblock URL \url{http://dx.doi.org/10.1287/opre.26.1.22}.

\bibitem[Leung et~al.(2007)Leung, Li, and Pinedo]{leung2007scheduling}
Joseph Y-T Leung, Haibing Li, and Michael Pinedo.
\newblock Scheduling orders for multiple product types to minimize total
  weighted completion time.
\newblock \emph{Discrete Applied Mathematics}, 155\penalty0 (8):\penalty0
  945--970, 2007.

\bibitem[Levey and Rothvoss(2016)]{levey20161}
Elaine Levey and Thomas Rothvoss.
\newblock A (1+ epsilon)-approximation for makespan scheduling with precedence
  constraints using lp hierarchies.
\newblock In \emph{Proceedings of the 48th Annual ACM SIGACT Symposium on
  Theory of Computing}, pages 168--177. ACM, 2016.

\bibitem[Li(2017)]{Li17}
Shi Li.
\newblock Scheduling to minimize total weighted completion time via
  time-indexed linear programming relaxations.
\newblock In \emph{Proceedings of the 2017 IEEE 58rd Annual Symposium on
  Foundations of Computer Science}, FOCS '17, 2017.

\bibitem[Manurangsi(2017)]{Manurangsi17}
Pasin Manurangsi.
\newblock Almost-polynomial ratio eth-hardness of approximating densest
  k-subgraph.
\newblock In \emph{Proceedings of the 49th Annual {ACM} {SIGACT} Symposium on
  Theory of Computing, {STOC} 2017, Montreal, QC, Canada, June 19-23, 2017},
  pages 954--961, 2017.
\newblock \doi{10.1145/3055399.3055412}.
\newblock URL \url{http://doi.acm.org/10.1145/3055399.3055412}.

\bibitem[Mastrolilli et~al.(2010)Mastrolilli, Queyranne, Schulz, Svensson, and
  Uhan]{mastrolilli2010minimizing}
Monaldo Mastrolilli, Maurice Queyranne, Andreas~S Schulz, Ola Svensson, and
  Nelson~A Uhan.
\newblock Minimizing the sum of weighted completion times in a concurrent open
  shop.
\newblock \emph{Operations Research Letters}, 38\penalty0 (5):\penalty0
  390--395, 2010.

\bibitem[Matou{\v{s}}ek(2002)]{matouvsek2002lectures}
Ji{\v{r}}{\'\i} Matou{\v{s}}ek.
\newblock \emph{Lectures on discrete geometry}, volume 212.
\newblock Springer Science \& Business Media, 2002.

\bibitem[Munier et~al.(1998)Munier, Queyranne, and Schulz]{MQS98}
Alix Munier, Maurice Queyranne, and Andreas~S. Schulz.
\newblock \emph{Approximation Bounds for a General Class of Precedence
  Constrained Parallel Machine Scheduling Problems}, pages 367--382.
\newblock Springer Berlin Heidelberg, Berlin, Heidelberg, 1998.
\newblock ISBN 978-3-540-69346-8.
\newblock \doi{10.1007/3-540-69346-7_28}.
\newblock URL \url{http://dx.doi.org/10.1007/3-540-69346-7_28}.

\bibitem[Qiu et~al.(2015)Qiu, Stein, and Zhong]{qiu2015minimizing}
Zhen Qiu, Cliff Stein, and Yuan Zhong.
\newblock Minimizing the total weighted completion time of coflows in
  datacenter networks.
\newblock In \emph{Proceedings of the 27th ACM symposium on Parallelism in
  Algorithms and Architectures}, pages 294--303. ACM, 2015.

\bibitem[Queyranne and Schulz(2006)]{QS06}
Maurice Queyranne and Andreas~S. Schulz.
\newblock Approximation bounds for a general class of precedence constrained
  parallel machine scheduling problems.
\newblock \emph{SIAM J. Comput.}, 35\penalty0 (5):\penalty0 1241--1253, May
  2006.
\newblock ISSN 0097-5397.
\newblock \doi{10.1137/S0097539799358094}.
\newblock URL \url{http://dx.doi.org/10.1137/S0097539799358094}.

\bibitem[Shafiee and Ghaderi(2017)]{shafiee2017improved}
Mehrnoosh Shafiee and Javad Ghaderi.
\newblock An improved bound for minimizing the total weighted completion time
  of coflows in datacenters.
\newblock \emph{arXiv preprint arXiv:1704.08357}, 2017.

\bibitem[Svensson(2010)]{svensson2010conditional}
Ola Svensson.
\newblock Conditional hardness of precedence constrained scheduling on
  identical machines.
\newblock In \emph{Proceedings of the forty-second ACM symposium on Theory of
  computing}, pages 745--754. ACM, 2010.

\bibitem[Varadarajan(2009)]{varadarajan2009epsilon}
Kasturi Varadarajan.
\newblock Epsilon nets and union complexity.
\newblock In \emph{Proceedings of the twenty-fifth annual symposium on
  Computational geometry}, pages 11--16. ACM, 2009.

\bibitem[Wagneur and Sriskandarajah(1993)]{wagneur1993openshops}
Edouard Wagneur and Chelliah Sriskandarajah.
\newblock Openshops with jobs overlap.
\newblock \emph{European Journal of Operational Research}, 71\penalty0
  (3):\penalty0 366--378, 1993.

\bibitem[Wang and Cheng(2007)]{wang2007customer}
Guoqing Wang and TC~Edwin Cheng.
\newblock Customer order scheduling to minimize total weighted completion time.
\newblock \emph{Omega}, 35\penalty0 (5):\penalty0 623--626, 2007.

\bibitem[Yang(1998)]{yang1998scheduling}
Jaehwan Yang.
\newblock \emph{Scheduling with batch objectives}.
\newblock PhD thesis, Ohio State University, 1998.

\end{thebibliography}
